\documentclass[11pt]{article}
\usepackage[margin=1in]{geometry}

\newcommand {\ignore} [1] {}
\newcommand{\remove}[1]{}

\usepackage[utf8]{inputenc} %
\usepackage{booktabs}       %
\usepackage{amsfonts}       %
\usepackage{nicefrac}       %
\usepackage{microtype}      %
\usepackage{parskip}        %
\usepackage[dvipsnames]{xcolor} %
\usepackage{amsmath, amsthm, mathtools, dsfont}
\usepackage{txfonts} %
\usepackage{color,graphicx}
\usepackage{url,hyperref}
\usepackage{enumerate}
\usepackage{array}

\usepackage{amssymb}

\usepackage[ruled,vlined]{algorithm2e}
\usepackage{algorithmic,multicol}

\newcommand{\accset}{A}

\newcommand{\optmatch}{T^{\textnormal{opt}}}

\def\reals{\mathbb{R}}
\def\nats{\mathbb{N}}

\newcommand{\Mod}[1]{\ (\mathrm{mod}\ #1)}

\def\Exp{\mathbb{E}}
\def\Prob{\mathbb{P}\mathrm{r}}

\newcommand{\R}{\mathbb{R}}

\newcommand{\N}{\mathbb{N}}

\renewcommand{\th}{\textsuperscript{th} }

\renewcommand{\phi}{\varphi}

\newcommand{\EqComment}[1]{\text{\emph{(#1)}}}

\newcommand{\EXP}{\mathbb{E}}
\newcommand{\E}{\EXP}

\newcommand{\cbr}[1]{\left\{\,#1\,\right\}}
\newcommand{\floor}[1]{\left\lfloor #1 \right\rfloor}

\usepackage{color-edits}
\usepackage[dvipsnames]{xcolor}
\addauthor{kb}{cyan}   %
\addauthor{jo}{blue}   %
\addauthor[Danny]{dm}{Green}   %

\newtheorem{theorem}{Theorem}

\newtheorem{lemma}[theorem]{Lemma}
\newtheorem{fact}[theorem]{Fact}

\newtheorem{definition}[theorem]{Definition}

\SetKwComment{Comment}{ // }{}

\SetKw{Continue}{continue}

\title{
Beating Competitive Ratio $4$ for Graphic
Matroid Secretary
\footnote{Submitted to STOC 2025 on November 4, 2024.}
}

\author{
Kiarash Banihashem 
\thanks{University of Maryland, College Park, Maryland, USA. \texttt{kiarash@umd.edu}.}
\and
MohammadTaghi Hajiaghayi 
\thanks{University of Maryland, College Park, Maryland, USA. Email: \texttt{hajiagha@umd.edu}.}
\and
Dariusz R. Kowalski
\thanks{Augusta University, Augusta, Georgia, USA.
Email: \texttt{dkowalski@augusta.edu}.
}
\and
Piotr Krysta
\thanks{Augusta University, Augusta, Georgia, USA.
Email: 
\texttt{pkrysta@augusta.edu}.
Piotr Krysta is also affiliated with Computer Science Department, University of Liverpool, U.K.
}
\and
Danny Mittal
\thanks{University of Maryland, College Park, Maryland, USA.
Email: \texttt{dannymittal@gmail.com}.}
\and 
Jan Olkowski
\thanks{University of Maryland, College Park, Maryland, USA.
Email: \texttt{olkowski@umd.edu}.}
}

\date{}
\usepackage{csquotes}
\begin{document}
\maketitle

\begin{abstract}
One of the classic problems in online decision-making is the {\em secretary problem}, where the goal is to hire the best secretary out of $n$ rankable applicants or, in a natural extension, to maximize the probability of selecting the largest number from a sequence arriving in random order.
Many works have considered generalizations of this problem where one can accept multiple values subject to a combinatorial constraint. 
The seminal work of  Babaioff, Immorlica, Kempe, and Kleinberg (SODA'07, JACM'18) proposed the {\em matroid secretary conjecture}, suggesting that 
there exists an $O(1)$-competitive algorithm 
for the matroid constraint,
and many works since have attempted to obtain algorithms for both general matroids and specific classes of matroids.
The ultimate goal of these results is to obtain an $e$-competitive algorithm, and the {\em strong matroid secretary conjecture} states that this is possible for general matroids.

One of the most important classes of matroids is the
{\em graphic matroid}, where
a set of edges in a graph is deemed independent if it contains no cycle. 
Given the rich combinatorial structure of graphs, obtaining algorithms for these matroids is often seen as a good first step towards solving the problem
for general matroids.
For matroid secretary, Babaioff et al. (SODA'07, JACM'18) first studied graphic matroid case and  obtained a $16$-competitive algorithm.
Subsequent works have improved the competitive ratio, most recently to $4$ by Soto, Turkieltaub, and Verdugo (SODA'18).

In this paper, we break the $4$-competitive barrier for the problem, obtaining a new algorithm with a competitive ratio of $3.95$. 
For the special case of simple graphs (i.e., graphs that do not contain parallel edges) we further improve this to $3.77$. 
Intuitively, solving the problem for simple graphs is easier as they do not contain cycles of length two. 
A natural question that arises is whether we can obtain a ratio arbitrarily close to $e$ by assuming the graph has a large enough girth.

We answer this question affirmatively,
proving that one can obtain a competitive ratio arbitrarily close to $e$ even for constant values of girth, providing further evidence for the strong matroid secretary conjecture.
We further show that this bound is tight: for any constant $g$, one cannot obtain a competitive ratio better than $e$ even if we assume that the input graph has girth at least $g$.
To our knowledge, such a bound was not previously known even for simple graphs.

\end{abstract}

\section{Introduction}

In the past two decades there has been a renewed interest in online item selection problems
where a sequence of items arrive one by one, revealing their \emph{weight}, and a decision maker needs to irrevocably decide whether or not to accept each item as it arrives.
The goal is to maximize the total accepted weight, subject to a feasibility constraint on the chosen items.
An algorithm's performance is typically measured by its \emph{competitive ratio}, which compares the algorithm’s total weight with the offline optimum---the total weight achievable if all item weights were known in advance.
These problems are appealing both from a mathematical perspective, as they are concise models for online decision making, 
and from an economical perspective, as they have close connection to pricing and auction theory~\cite{HajiaghayiKP04, hajiaghayi2007automated, babaioff2007matroids, chawla2009sequential}.

In the absence of any information about future items, the problem is essentially hopeless even with two items as the first item may be either significantly heavier or significantly lighter than the second item, and the decision maker has no way of deciding which is the case.
As such, most works make distributional assumptions on either the weight of the arriving items, or their arrival order.
The former class of problems are generally referred to as \emph{prophet inequalities} while the latter are known as \emph{secretary problems}.
Numerous works have studied secretary problems for a large class of combinatorial constraints~\cite{KorulaP09, dimitrov2008competitive, KesselheimRTV13, jaillet2013advances, 
dinitz2014matroid, rubinstein2016beyond, SotoTV18, EzraFGT22, 0002PZ24} and objective functions~\cite{BHZ13,FeldmanZ18}, and considered close variants of these problems such as the prophet-secretary problem~\cite{EHLM17}. 

Perhaps the most important open question in the area of online decision making is the \emph{matroid secretary proplem} posed by the seminal work of Babaioff, Immorlica, Kempe, and Kleinberg (SODA'07,JACM'18)~\cite{babaioff2007matroids,BabaioffIKK18}. 
In the matroid secretary problem, items arrive in a random order, and each item corresponds to an element in a matroid 
$M=(E,I)$,where 
$E$ is the set of elements and 
$I$ is the collection of independent sets in matroid $M$.  Upon arrival, the weight of each item is revealed, and the decision maker must immediately decide whether to accept or reject it. The goal is to maximize the total weight of accepted items, with the constraint that the selected items form an independent set in the matroid.
The matroid secretary conjecture~\cite{babaioff2007matroids,BabaioffIKK18} states that there exists a constant-competitive ratio algorithm for this problem, and \emph{strong matroid secretary conjecture} (e.g., see \cite{SotoTV18}) states that there exists an $e$-competitive algorithm.
Many works have studied the problem for both general matroids and specific cases.

In this paper, we focus on the specific case of graphic matroids. 
In this case, the arriving items corresponds to the edges of a graph and the goal is to accept a set of edges that do not contain a cycle.
Solving problems for graphic matroids is often viewed as a promising first step toward addressing arbitrary matroids, as graphs possess rich combinatorial structures and graphic matroids, along with linear matroids, are among the most intuitive examples of non-trivial matroids.
Many counterexamples for candidate matroid secretary algorithms are, in fact, graphic matroids~\cite{babaioff2007matroids,BahraniBSW21}. It is common to illustrate the main ideas behind general algorithms by showing their behavior on this specific case~\cite{svensson_matroid_secretary_2016}. Recent work has also explored whether techniques for graphic matroids can be extended to general matroids~\cite{AbdolazimiKKG23}.

The seminal paper of Babaioff et al. originally studied the graphic matroid secretary problem, obtaining a $16$-competitive algorithm for the problem. Babaioff et. al.~\cite{babaioff2009secretary} designed a $3e$-competitive algorithm. 
The competitive ratio was later improved to $2e$ by Korula and Pal~\cite{KorulaP09}
and later to $4$ by Soto, Turkieltaub, and Verdugo~\cite{SotoTV18}.
Whether or not the ratio can be improved has been open at least since 2018.

\subsection{Our results and techniques}
In this paper we obtain an algorithm with competitive ratio $3.95$ for the problem, breaking the $4$-competitive barrier.
We further improve this result for the specific case of \emph{simple graphs},  i.e., graphs that do not have parallel edges.
Intuitively, simple graphs represent an easier special case as they don't have any cycles of length $2$; since the algorithm is forbidden from accepting edges that form a cycle, the lack of $2$-cycles gives the algorithm more freedom to accept edges.
Our main result is the following theorem.
\begin{theorem}
\label{mainthm}
    There exists a $3.95$-competitive algorithm for the graphic matroid secretary problem. Furthermore, if the input graph is assumed to be simple, there exists an algorithm with competitive ratio~$3.77$.
\end{theorem}
The basic approach for this result is to, at each step, compute a set of outgoing edges such that each node has exactly one outgoing edge. More exactly, at each step we compute a maximum spanning forest, then direct the edges in this forest toward an arbitrary root; each node then has a unique outgoing edge. We then argue that if the algorithm only accepts an edge for whose endpoints it has not previously accepted an outgoing edge, then the algorithm's set of taken edges will always be independent.

Given this, the key aspect of Algorithm \ref{alg:graph} that allows it to obtain an improved competitive ratio is a slightly stronger condition used to determine whether an edge is taken. Specifically, while for one endpoint we only demand that we have not taken an outgoing edge, for the other endpoint we demand that we have not even seen an outgoing edge. We can then show an increased probability for the former condition being satisfied for a currently considered by using the fact that the latter condition may not be satisfied by a previously seen outgoing edge, causing said edge to not be taken. Algorithm \ref{alg:graph} furthermore makes use of random choice in determining which endpoint to apply the stronger condition to, which is crucial for handling the case of duplicate edges. When the input graph is guaranteed to be simple, this random choice is unnecessary -- removing this random choice gives Algorithm \ref{alg:simple-graph} which allows us to obtain an even lower competitive ratio.

Motivated by the improvement for simple graphs,
we additionally study, for the first time, the landscape of the graphic matroid secretary problem for graphs of high \emph{girth}, where we recall that the girth of a graph is the length of its shortest cycle.
We show that when the graph has large (but constant) girth, one can obtain a competitive ratio arbitrarily close to $e$. Formally, we prove the following theorem. 
\begin{theorem}\label{highgirththm}
For any graph $G$ with girth at least $g$,
there exists an algorithm being $\frac{1}{e} \big(1 - o_g(1)\big)$-competitive.
\end{theorem}
The result builds on yet another property of the combinatorial structure of the set of outgoing edges introduced in the proof of Theorem \ref{mainthm}. We first observe that the fact that every vertex has at most one outgoing edge implies that, in a graph induced by the outgoing edges, each edge belongs to at most one cycle. On the other hand, we prove that if Algorithm \ref{alg:graph} were to accept outgoing edges without respecting the graphic matroid independence condition, it would yield a higher acceptance probability of $\frac{1}{e}$ for an edge from the maximum independent set. These observations lead to a natural approach: accepting an outgoing edge only with a certain probability such that, for every cycle of length $g$, the probabilities of taking each edge of this cycle are equal. Intuitively, as the length of the shortest cycle increases, this probability tends towards $\frac{1}{e}$. We refer the reader to Section~\ref{sec:high-girth} for more details.

The best competitive ratios we attain over the three algorithms we introduce are listed in Table \ref{tab:competitive_ratios} for girths less than $10$.
\begin{table}[h!]
    \centering
    \noindent
    \begin{minipage}{0.45\textwidth}
        \centering
        \begin{tabular}{@{} c c c @{}}
            \toprule
            \textbf{Girth (\( g \))} & \textbf{Competitive Ratio} & \textbf{Algorithm} \\
            \midrule
        2 & 3.95 & Algorithm \ref{alg:graph} \\
        3 & 3.77 & Algorithm \ref{alg:simple-graph} \\
        4 & 3.77 & Algorithm \ref{alg:simple-graph} \\
        5 & 3.76 & Algorithm \ref{alg:high_girth} \\
            \bottomrule
        \end{tabular}
    \end{minipage}%
    \hspace{0.05\textwidth}%
    \begin{minipage}{0.45\textwidth}
        \centering
        \begin{tabular}{@{} c c c @{}}
            \toprule
            \textbf{Girth (\( g \))} & \textbf{Competitive Ratio} & \textbf{Algorithm} \\
            \midrule
        6 & 3.61 & Algorithm \ref{alg:high_girth} \\
        7 & 3.50 & Algorithm \ref{alg:high_girth} \\
            8 & 3.42 & Algorithm \ref{alg:high_girth} \\
            9 & 3.35 & Algorithm \ref{alg:high_girth} \\
            \bottomrule
        \end{tabular}
    \end{minipage}
    \caption{The best competitive ratio we obtain for graphic matroid secretary when the input graph is restricted to have girth at least $g$ for $g < 10$. $g = 2$ corresponds to multigraphs, while $g = 3$ corresponds to simple graphs. As $g$ approaches infinity, the competitive ratio approaches $e$.
    }
    \label{tab:competitive_ratios}
\end{table}

Perhaps more surprisingly, we show that this is tight and that no algorithm can obtain a competitive ratio better than $e$, even if the graph is assumed to have a high girth.
To our knowledge, such a lower bound was not previously known, even for the special case of simple graphs, and we believe our techniques are of independent interest for similar online arrival problems.

\begin{theorem}\label{lowerboundthm}
For any $g \in \N$, there does not exist an algorithm for the graphic matroid secretary problem on graphs of girth $\geq g$ that obtains competitive ratio less than $e$.
\end{theorem}

The two results fully characterize the landscape of the secretary problem in the high-girth setting, essentially showing that the problem becomes
as easy as the single secretary problem in the limit.

To prove the lower bound in Theorem \ref{lowerboundthm}, we first show that weighted single secretary, also known as a cardinal secretary, is hard ``on average'' in the following sense. We demonstrate the existence of a finite distribution over the instances of this problem such that even if adversary samples instances, as opposed to choosing them in a worst-case manner, from that distribution, it is still hard for any algorithm to choose the maximum weight element.

We then construct a high girth bipartite graph based on Ramanujan graphs, where vertices of one bipartition have same degree $d$. The hard distribution of instances of weighted single secretary of size $d$ can now be sampled independently and embedded on the edges incident to each degree-$d$ vertex on one side of the constructed bipartite graph. Given an algorithm that performs well on this graph, we obtain an algorithm performing well on the original distribution of weighted single secretary by simulating the first algorithm and mimicking the choices that it makes on a single degree-$d$ vertex.

Our construction of the input graph in addition to high-girth Ramanujan graphs uses the probabilistic method. Our construction of the sets of hard weights is based on an extension of a similar argument from an infinite to a finite version of Ramsey theorem. We employ zero-sum games duality argument to show an existence of the final probabilistic distribution over the instances.

\section{Related work}

Many other derivations and specific cases of the general matroid secretary problem have been given attention over years. Hajiaghayi, Kleinberg and
Parkes~\cite{HajiaghayiKP04} had first introduced the
{\em multiple-choice value version} of the problem, aka the uniform matroid secretary problem, in which the goal
is to maximize the expected sum of the at most $k$ selected numbers.
Kleinberg~\cite{kleinberg2005multiple} later presented a tight
$(1-O(\sqrt{1/k})^{-1})$-competitive algorithm for the $k$ uniform secretary resolving an open problem of~\cite{HajiaghayiKP04}. Transversal matroids have been considered first in~\cite{babaioff2007matroids} who gave $4d$-competitive algorithm, where $d$ is the degree of the transversal matroid. This was improved by Dimitrov and Plaxton~\cite{dimitrov2008competitive} who showed a ratio of $16$ for all transversal matroids.  
For the laminar matroids, a long line of work led to $3\sqrt{3}e$ competitive ratio~\cite{im2011secretary, jaillet2013advances, ma2016simulated}.
The challenging general class of regular matroids was proven to admit $9e$-competitive algorithm~\cite{dinitz2014matroid}.

Other generalizations of the secretary problem such as the submodular variant have been initially studied by the Bateni, Hajiaghayi, and ZadiMoghaddam~\cite{BHZ13} and  Gupta, Roth, Schoenebeck, and
Talwar~\cite{DBLP:conf/wine/GuptaRST10}. The connection between the secretary problem and online auction
mechanisms has been explored by 
Kesselheim et al.~\cite{KesselheimRTV13}, 
who
give a
$e$-competitive solution to the online bipartite weighted matching problem.

The \emph{prophet secretary} problem is another well-studied variant of the secretary problem, closely related to prophet inequalities. In the \emph{prophet inequality} setting, introduced by Krengel and Sucheston~\cite{krengel1977semiamarts,krengel1978semiamarts}, we know the distributions of \( n \) arriving items and aim to maximize the expected ratio of the selected item’s value to the sequence maximum, with a tight competitive ratio of \(2 \). Research connecting prophet inequalities and online auctions, initiated by Hajiaghayi, Kleinberg, and Sandholm~\cite{hajiaghayi2007automated}, led to follow-up studies such as Alaei, Hajiaghayi, and Liaghat’s work~\cite{AHL13} on the bipartite matching variant of prophet inequality, also achieving a competitive ratio of \( 2 \)~\cite{AHL13}. Feldman et al.~\cite{feldman2015combinatorial} expanded the problem to combinatorial auctions with multiple buyers, achieving the same bound through a posted pricing scheme, and Kleinberg and Weinberg~\cite{KW-STOC12} extended this result to matroids with a \(2 \)-competitive algorithm. 
The prophet secretary model, introduced by Esfandiari, Hajiaghayi, Liaghat, and Monemizadeh~\cite{EHLM17}, assumes a random arrival order and known item distributions. They designed an algorithm achieving a competitive factor of \( (\frac{e}{e-1}) \), which has proven challenging to improve. However, Azar et al.~\cite{ACK18} and Correa et al.~\cite{CorreaSZ19,DBLP:journals/mp/CorreaSZ21} improved this bound to $(1-1/e + 1/30)^{-1} \approx 1.502$. For the single-item i.i.d. case, Abolhasani et al.~\cite{abolhassani2017beating} achieved a \( 1.37 \)-competitive ratio, later improved to \( 1.342 \) by Correa et al.~\cite{correa2017posted}. Recently, Peng and Tang~\cite{PT22} achieved a $1.379$-competitive algorithm for the {\em free order} case, however finding the tight competitive bound for the general prophet secretary problem remains an important open problem.

Last but not least, extensions beyond matroids for the secretary problem have also been studied. For example, Kleinberg and Weinberg~\cite{KW-STOC12} provided an $O(p)$-competitive algorithm for the intersection of \( p \) matroids, later generalized to polymatroids by Dütting and Kleinberg~\cite{dutting2015polymatroid}. Rubinstein~\cite{rubinstein2016beyond} and Rubinstein and Singla~\cite{RS-SODA17} considered prophet inequalities and secretary problems in arbitrary downward-closed set systems. For such settings, Babaioff et al.~\cite{babaioff2007matroids} proved a lower bound of \( \Omega(\log n \log \log n) \), and further studies have explored combinatorial optimization applications~\cite{DEHLS17,garg2008stochastic,gobel2014online,Meyerson-FOCS01}.

\section{Preliminaries}
\paragraph{Graph Notation.} In this paper, we assume all graphs are undirected, weighted, and may contain multiple edges between the same pair of vertices. Specifically, a graph is defined as a triple \( G = (V, E, w : E \rightarrow \mathbb{R}) \), where \( V \), with \( |V| = n \), is the set of vertices; \( E \), with \( |E| = m \), is the multiset of edges; and \( w \) assigns weights to the edges. We assume graphs do not contain loops (i.e., edges connecting a vertex to itself\footnote{As explained in the following paragraph, such edges are irrelevant in the context of graphical matroids.}). Given any vertex $u \in V$, we denote the degree of $u$ by $deg_G(u)$. A graph is called \textit{simple} if there is at most one edge connecting any pair of vertices. For a graph \( G \), we denote \( g \) as the \textit{girth} of \( G \), representing the length of its shortest cycle. In graphs with multiple edges, a cycle is defined as any multiset of edges \( C = \{(a_{1}, b_{1}), \ldots, (a_{k}, b_{k})\} \), for \( k \geq 2 \), such that \( \forall_{1 \leq i \leq k-1} \, b_{i} = a_{i+1} \) and \( b_{k} = a_{1} \). For an integer \( g \geq 2 \), let \( \mathcal{G}_{g} \) denote the set of all graphs with girth at least \( g \). We also denote $[n] =\{1,2,\ldots,n\}$.

\paragraph{Matroids.} A \textit{matroid} \( M = (E, I) \) is a combinatorial structure that generalizes independence. It consists of a finite set \( E \) and a collection \( I \) of independent subsets of \( E \) satisfying: (1) the empty set is independent, (2) any subset of an independent set is also independent, and (3) if one independent set is larger than another, an element from the larger set can extend the smaller one while preserving independence. These properties capture the concept of independence, making matroids useful for modeling optimization problems where we seek a maximum-weight or maximum-cardinality independent subset of elements.

\paragraph{Graphic Matroids.} A \textit{graphic matroid} \( \mathcal{M}(G) \) associated with a graph \( G \) is defined as follows: the elements of the matroid are the edges in the multiset \( E \), and the independent sets are all acyclic subgraphs (i.e., subsets of edges) of \( G \). The weight of an independent set is the sum of the weights of its edges.

\paragraph{Problems.} Consider a graphic matroid \( \mathcal{M}(G) \) associated with a \textit{multigraph} \( G = (V, E, w : E \rightarrow \mathbb{R}) \). In the online secretary problem on graphical matroids, the elements of \( E \) are presented to the algorithm in a random order, chosen uniformly from all possible permutations of the multiset \( E \). Elements arrive one at a time, effectively creating \( m \) time steps during the algorithm’s execution. Upon the arrival of an element, the algorithm must decide whether to accept or reject it, with the constraint that an element can only be accepted if it forms an independent set with the already accepted elements. Decisions are irrevocable.

The objective is to design an algorithm that maximizes the expected sum of the weights of the accepted elements, referred to as the algorithm’s gain. For an algorithm \( ALG \), we denote its (random) gain by \( ALG \) and its expected gain by \( \mathbb{E}(ALG) \).

Let \( OPT \) denote the maximum weight of an independent set in the matroid \( \mathcal{M}(G) \). We say an algorithm \( ALG \) is \textit{\(\alpha\)-competitive} (or has an \(\alpha \) competitive ratio, for some \(\alpha \geq 1 \)) for a family of matroids if
\[
\alpha \cdot \mathbb{E}(ALG) \geq OPT
\]
for all matroids in that family.

\section{Improved bounds for graphic matroid secretary: proof of Theorem~\ref{mainthm}}
In this section, we prove Theorem \ref{mainthm} by proving Theorems \ref{thm:graph-algorithm} and \ref{thm:simple-graph-algorithm}. To prove the former, we present an algorithm that attains a competitive ratio less than $3.95$ for the graphic matroid secretary problem, surpassing a result of Soto, Turkieltaub, and Verdugo \cite{SotoTV18} attaining a competitive ratio of $4$ that previously stood for six years as the best result achieved for graphic matroid secretary; this algorithm is described in Section \ref{section:multigraph}. To prove the latter, we present a slight modification of our algorithm that attains a competitive ratio less than $3.77$ for the graphic matroid secretary problem in the special case of simple graphs; this algorithm is described in Section \ref{section:simple}.

\subsection{Improvement for general graphs}
\label{section:multigraph}

\paragraph{Overview.}

\begin{algorithm}[!ht]
  Let $E'$ be the first $m'$ edges \hfill \Comment{set of observed edges}
  $\accset \gets \emptyset$ \hfill\Comment{set of accepted edges}
  
  $\forall v\in V$: $seen\_outgoing(v)\gets False$\\
  $\forall v\in V$: $taken\_outgoing(v)\gets False$\\
  \For{$t \in \cbr{m' + 1, \dots, m}$} {
    Let $e_t$ be the edge arriving in time $t$\\
    Add $e_t$ to $E'$ \\
    $\optmatch_t \gets $ maximum weighted directed
    forest on $G[E']$\\
    $\forall v \in V: outgoing(v) \gets $ the edge directed away from $v$ in $\optmatch_t$ if it exists\\
    $\forall v \in V$ such that $outgoing(v) = null : outgoing(v) \gets $ a unique edge in $E'$ not in $\optmatch_t$\\
    \If{$e_{t} \in \optmatch_{t}$} { 
    $e_{t} \gets (u, v)$ \hfill
    \Comment{where $e_t$ is directed from $u$ to $v$ in $\optmatch_{t}$
    }
    $(a, b) \gets $ either $(u, v)$ or $(v, u)$ with equal probability
    
    \If{$seen\_outgoing(a) = False$ and $taken\_outgoing(b) = False$}
    {
      $taken\_outgoing(u) \gets True$ \\
      add $e_{t}$ to $\accset$
    }
    
  }
  \If{$outgoing(w) = e_t$ for some $w \in V$} {
    $seen\_outgoing(w) \gets True$
  }
  }
  \caption{New algorithm for graphic matroid secretary.}
  \label{alg:graph}
\end{algorithm}

The pseudocode of the algorithm is provided in Figure~\ref{alg:graph}. The algorithm consists of two phases. During the first $m'$ steps (where $m'$ depends only on $m$), edges are only observed without being taken. In later steps, let $e$ be the presented edge. The algorithm computes a maximum spanning forest $\optmatch_t$ of all edges seen so far. This forest is directed away from the root, so that every edge is directed and each node has at most one outgoing edge. We additionally compute an array of \textit{outgoing} edges: for each node $v$, if $v$ has an outgoing edge in $\optmatch_t$, then that it is its outgoing edge; if it does not, then we choose an arbitrary edge not already assigned to be the outgoing edge of another vertex \footnote{Note that if $m$ is too small (specifically, if $m' + 1 < n$) then there may not exist sufficiently many edges for this to be possible. However, if this is the case, we can simply mix in "dummy edges" that the algorithm can treat the same as real edges in order to cause $m' + 1$ to be at least $n$. The details of this mixing in are described in the proof of Theorem \ref{thm:graph-algorithm}.} and declare that to be the outgoing edge from $v$. Note that said edge does not even have to have $v$ as an endpoint. Also note that crucially, the choice of edges declared to be outgoing for the vertices that need it depends on the set of edges observed so far but not their order.

We only consider $e$ if it is in $\optmatch_t$. In this case, it points from some $u$ to some $v$, and so is the outgoing edge for $u$. In order to take $e$, we randomly order $u, v$ as $a, b$ and impose the following constraints on $e$:
\begin{itemize}
    \item No edge has been taken that was the outgoing edge for $b$ when it was presented.
    \item No edge has appeared that was the outgoing edge for $a$ when it was presented.
\end{itemize}
If $e$ satisfies both constraints then it is taken. Otherwise, the algorithm still notes that $e$ could have been selected for $u$.

The distinct constraints used for $a, b$ are key -- the stronger condition applied to endpoint $a$ allows us to lower bound the probability that an edge is blocked from being taken, which can then be used to show an increased probability that the weaker condition for endpoint $b$ is satisfied for an edge that we would like to be taken. The randomness in ordering $u, v$ into $a, b$ is also essential. To see why, consider the specific example of an edge $e_t \in OPT$ with endpoints $u, v$ that is presented at step $t$, and consider some step $j < t$. Suppose that $\optmatch_j$ contains an edge $e^u$ outgoing from $u$ to a third node $x$ and an edge $e^v$ outgoing from $v$ to a fourth node $y$. We would like to argue that in the case that either $e^u$ or $e^u$ is the edge $e_j$ presented at step $j$, there is a possibility that some outgoing edge from $x$ or $y$ respectively had already been presented prior to step $j$, meaning that $e^u$ or $e^v$ could not have been taken.

\begin{figure}[t!]
    \centering
    \includegraphics[width=1\textwidth]{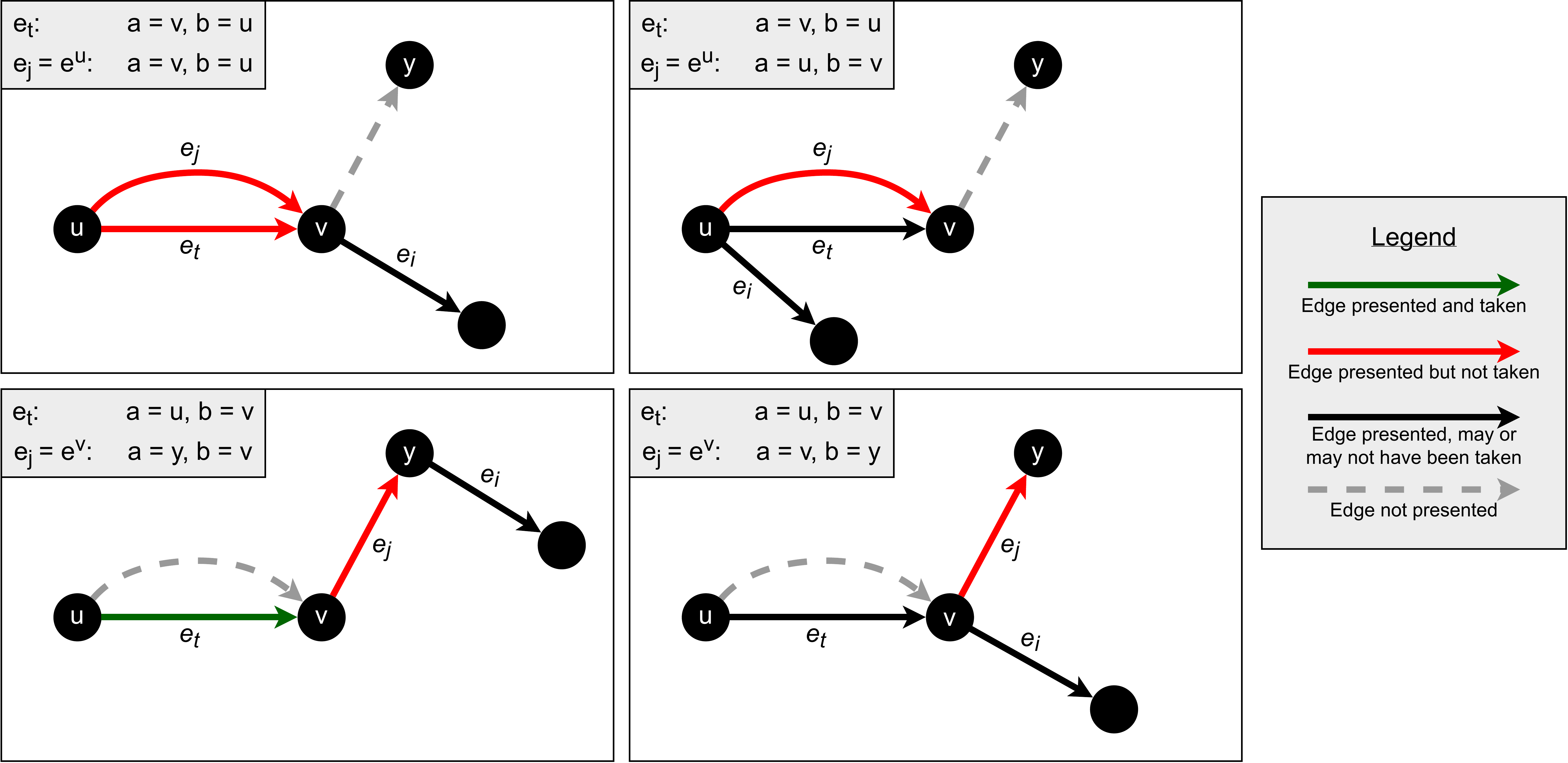} 
    \caption{The described "worst-case" example for Algorithm \ref{alg:graph}. At step $t$, we are presented with an edge $e_t$ in the optimum solution. $e_t$ may be blocked from being taken due to an earlier step $j$. We suppose that the outgoing edge $e^u$ from $u$ in $\optmatch_j$ goes to $v$ while the outgoing edge $e^v$ from $v$ in $\optmatch_j$ goes to a third vertex $y$; these are the only possible values of $e_j$ with a potential to block $e_t$. All of $e_t, e^u, e^v$ are depicted. Additionally depicted is an edge $e_i$ outgoing from the endpoint of $e_j$ selected as $b$ -- this edge being presented at step $i$ would block $e_j$ from being taken. Each of the four images depicts one equally likely possibility for the random choice of $(a, b)$ in the case of $e_t$ as well as in the case of $e_j$, where $e_j$ is assumed to be the edge outgoing from $b$ in $\optmatch_j$ (as this is the only case relevant for an increased probability arising from $e_j$ not being taken). Only in one case is $e_t$ guaranteed to be taken (assuming steps other than those depicted do not block $e_t$ from being taken).}
    \label{fig:four_cases}
\end{figure}

This means that, for example, if $e^u$ were the edge presented at step $j$, we would desire that, in the case of edge $e_t$, $b = u$, and in the case of edge $e_j = e^u$, $a = x$ (and so $b = u$), so that $e^u$ could be blocked by an edge outgoing from $x$ simply being presented, while $e_t$ is then not blocked because $e^u$ was not actually taken. This suggests that we could simply always take $b = u$ and $a = v$, applying the stronger condition to the endpoint which the edge is directed towards.

However, as in general graphs it is possible for multiple edge to have the same endpoints, it is possible that we in fact have $x = v$, in which case an edge outgoing from $x$ would block $e_t$ itself. This is illustrated in the top-left quadrant of Figure \ref{fig:four_cases}. Therefore, in order to guarantee an improvement, we crucially allow for the possibility that in the case of edge $e_t$, $b = v$, and in the case of edge $e_j = e^v$, $a = y$ -- if $e^u$ has the same endpoints as $e$, then $e^v$ cannot as $\optmatch_j$ is a spanning tree, and so $y$ will necessarily be distinct from $u$, meaning that there is a possibility for $e^v$ to be blocked by an edge outgoing from $y$ while still allowing $e_t$ to be taken. This is illustrated in the bottom-left quadrant of Figure \ref{fig:four_cases}.

\paragraph{Analysis.}

We aim to prove the following theorem:
\begin{theorem}\label{thm:graph-algorithm}
    There exists a $3.95$-competitive algorithm for graphic matroid secretary.
\end{theorem}
We first show that any set of edges accepted by Algorithm \ref{alg:graph} is independent, meaning that the algorithm is a valid graphic matroid secretary algorithm.
\begin{lemma}\label{lemma:graph-algorithm-independent}
    The set $\accset$ of accepted edges in Algorithm \ref{alg:graph} is always an independent set in the graphic matroid.
\end{lemma}
\begin{proof}
    Proof by contradiction. Suppose that there exists a cycle $C$ contained within $A$. Recall that for each edge we take, we declare it to be selected by one of its endpoints, and that an edge $(u, v)$ cannot be taken if for either $u$ or $v$ we have taken an edge outgoing from that vertex either (the exact condition is stronger, as for one of the endpoints we demand that we cannot have \textit{seen} an edge outgoing from that vertex). The two statements below follow:
    \begin{itemize}
        \item Each taken edge is outgoing from a unique node, which must be one of its endpoints, meaning that if $C$ has length $l$, then because $C$ contains $l$ edges whose endpoints are all in $C$, for all $l$ vertices in $C$ we must have taken an outgoing edge from that vertex.
        \item Let $e = (u, v)$ be the edge in $C$ that was presented at the latest time, and let this time be $t$. Prior to step $t$, we had not taken an edge outgoing from either $u$ or $v$. Thus, following step $t$, for one of $u, v$ we have taken the outgoing edge $e$, and for the other we still have not taken an outgoing edge; as a consequence, we never took an outgoing edge from that vertex in $C$.
    \end{itemize}
    These statements contradict each other, and so the proof is complete.
\end{proof}

We now proceed to demonstrate that Algorithm \ref{alg:graph} has the desired approximation factor for an appropriate choice of $m'$. We first define quantities $f, g, h$ below that will be crucial for the analysis:
\begin{definition}
    Define $f_{u, S}(i)$ to be the probability that after step $i$, $seen\_outgoing(u) = False$, given that $S$ is the set of edges that are presented in steps $1, \ldots, i$.
\end{definition}
\begin{definition}
    Define $g_{u, v, S}(i)$ to be the probability that after step $i$, both $seen\_outgoing(u) = False$ and $seen\_outgoing(v) = False$, given that $S$ is the set of edges that are presented in steps $1, \ldots, i$.
\end{definition}
\begin{definition}
    Define $h_{u, v, S}(i)$ to be the probability that after step $i$, if we take $(a, b)$ to be either $(u, v)$ or $(v, u)$ with equal probability, then $seen\_outgoing(a) = False$ and $taken\_outgoing(b) = False$, given that $S$ is the set of edges that are presented in steps $1, \ldots, i$.
\end{definition}
The condition used to define $h$ is the same condition checked to see whether an edge between $u$ and $v$ can be taken.

A key tool for our analysis is the following. Note that for any node $u$, $selected(u)$ is assigned when at some step $j$ we observe an edge $e$ such that $e$ is in $\optmatch_j$ and is directed from $u$. Importantly, $\optmatch_j$ depends only on the set of edges presented up to step $j$ and not their order. Therefore, conditioned on the set of edges that appears in steps $1, \ldots, j$, the probability that at step $j$ we are presented with the outgoing edge from $u$ is at most $\frac 1 j$ because there is exactly one such edge, and this edge, if it exists, is presented at step $i$ with probability $\frac 1 i$.

The above idea is applied in the following lemma:
\begin{lemma}\label{lemma:graph-algorithm-f}
    For all vertices $u$, all $i$ from $m'$ to $n$, and all $S$ with $|S| = i$, $f_{u, S}(i) = \frac {m'} i$.
\end{lemma}
\begin{proof}
    The proof is by induction. First note that as no edges are taken in the first $m'$ steps, $f_{u, S}(m') = 1 = \frac {m'} {m'}$. Then, for any $i > m'$, as previously described, $seen\_outgoing(u) = null$ after step $i$ iff $seen\_outgoing(u) = null$ before step $i$ and the edge presented at step $i$ is not the outgoing edge from $u$ based on $\optmatch_i$. As there is exactly one such edge, the probability that it appears at step $i$ is $\frac 1 i$, and so the probability that it does not is $1 - \frac 1 i$. $f_{u, S}(i)$ is therefore $1 - \frac 1 i$ multiplied by a weighted average of $f_{u, S'}(i - 1)$ over various sets $S'$ where $S'$ depends on what edge is presented at step $i$; as any such $f_{u, S'}(i - 1)$ is equal to $\frac m {i - 1}$, we can continue to state that $f_{u, S}(i) = (1 - \frac 1 i)\frac m {i - 1} = \frac {i - 1} i \cdot \frac m {i - 1} = \frac m i$ as desired.
\end{proof}
The next lemma uses the same idea to derive a similar expression for $g$:
\begin{lemma}\label{lemma:graph-algorithm-g}
    For all vertices $u, v$, all $i$ from $m'$ to $n$, and all $S$ with $|S| = i$, $g_{u, v, S}(i) \geq \frac {m'(m' - 1)} {i(i - 1)}$.
\end{lemma}
\begin{proof}
    The proof is by induction. First note that as no edges are taken in the first $m'$ steps, $g_{u, v, S}(m') = 1 = \frac {m'(m' - 1)} {m'(m' - 1)}$. Then, for any $i > m'$, as previously described, $seen\_outgoing(u), seen\_outgoing(v) = null$ after step $i$ iff $seen\_outgoing(u), seen\_outgoing(v) = null$ before step $i$ and the edge presented at step $i$ is not the outgoing edge from either $u$ or $v$ based on $\optmatch_i$. As there is exactly one such edge for each of $u, v$, the probability that such an edge appears at step $i$ is $\frac 2 i$. Arguing in a similar manner to the proof of Lemma \ref{lemma:graph-algorithm-f} we can say that $g_{u, v, S}(i)$ is $1 - \frac 2 i$ multiplied by the formula for any $g_{u, v, S'}(i - 1)$, giving the following:
    \begin{align*}
        g_{u, v, S}(i) = \left(1 - \frac 2 i\right) \frac {m'(m' - 1)} {(i - 1)(i - 1 - 1)} = \frac {i - 2} i \cdot \frac {m'(m' - 1)} {(i - 1)(i - 2)} = \frac {m'(m' - 1)} {i(i - 1)}.
    \end{align*}
\end{proof}
We now proceed to lower bound $h$; we again use the same idea but now additionally take advantage of the stronger condition applied to $a$ in the algorithm to derive a bound for $h$ superior than the one we derived for $g$ by lower bounding the probability that an edge that would have been selected for one of $u, v$ was not actually taken.

As the lower bound for $h$ is more complicated, we define its lower bound as an additional function $h_\delta$.
\begin{definition}\label{def:h-delta}
    Define $h_\delta(i)$ recursively by $h_a(m') = 1$ and $h_a(i) = (1 - \frac 2 i)h_a(i - 1) + \frac \delta i \cdot \frac {m'(m' - 1)} {(i - 1)(i - 2)} \left[1 - \frac {m'} {i - 1}\right].$
\end{definition}
We parameterize $h$ by $\delta$ as $h$ will reused later with a different value of $\delta$.
\begin{lemma}\label{lemma:graph-algorithm-h}
    For all vertices $u, v$, all $i$ from $m'$ to $n$, and all $S$ with $|S| = i$, $h_{u, v, S}(i) \geq h_{\frac 1 4}(i).$
\end{lemma}
\begin{proof}
    The proof is again by induction. Like before, we have $h_{u, v, S}(m') = 1 = h_{\frac 1 4}(m')$. Then, for any $i > m'$, as in the proof of Lemma \ref{lemma:graph-algorithm-g} we have that with probability $1 - \frac 2 i$ the edge that is presented at step $i$ is not directed from either of $u, v$; given this, the probability is a weighted average of terms $h_{u, v, S'}(i - 1)$, which we can lower bound by $h_{\frac 1 4}(i - 1)$. This contributes $(1 - \frac 2 i)h_{\frac 1 4}(i - 1)$ to our lower bound.

    In the case that an edge directed away from $u, v$ is selected, note that as $\optmatch_i$ is a spanning forest, at most one edge in $\optmatch_i$ can be between $u, v$, meaning that at least one of the edges directed from $u, v$ must connect to a third vertex $w$. There is a $\frac 1 i$ probability that said edge is presented at step $i$. Suppose without loss of generality that said edge goes from $u$ to $w$. There is a $\frac 1 2$ probability that we have $b = u$; independently, there is a $\frac 1 2$ probability that for the edge $(u, w)$ itself we have that $w$ is selected as the endpoint to which to apply the stronger condition. The probability that $seen\_outgoing(w) = True$ by step $i$ is then $(1 - f_{w, S'}(i - 1))$.
    
    We therefore have an additional probability of $\frac 1 i \cdot \frac 1 4 \cdot (1 - f_{w, S'}(i - 1))$ that the edge presented at step $i$ does not violate the conditions defining $h$. We additionally require that edges prior to step $i$ did not violate said conditions. Note that we are conditioned on the fact that $seen\_outgoing(w) = True$ by step $i$; we therefore impose the stronger condition that $seen\_outgoing(u) = seen\_outgoing(v) = False$ by step $i$, as the likelihood of this is only increased by the fact that $selected(w) \neq null$. We can thus use $g_{u, v, S''}(i - 1)$ as a lower bound for said probability.

    All of this combines to give us
    \begin{align*}
        h_{u, v, S}(i)
        &\geq \left(1 - \frac 2 i\right)h_{\frac 1 4}(i - 1) + \frac 1 i \cdot \frac 1 4 \cdot (1 - f_{w, S'}(i - 1)) \cdot g_{u, v, S''}(i - 1)\\
        &= \left(1 - \frac 2 i\right)h_{\frac 1 4}(i - 1) + \frac 1 i \cdot \frac 1 4 \cdot \left[1 - \frac {m'} {i - 1}\right] \cdot \frac {m'(m' - 1)} {(i - 1)(i - 2)} &\EqComment{Lemmas \ref{lemma:graph-algorithm-f} and \ref{lemma:graph-algorithm-g}}\\
        &= h_{\frac 1 4}(i). &\EqComment{Definition \ref{def:h-delta}}
    \end{align*}
\end{proof}
As mentioned before, the conditions defining $h_{u, v, S}(i)$ are identical to the conditions for taking an edge given that it is in $\optmatch_i$. As any edge in the overall optimum is in $\optmatch_i$, lower bounding the probability of such an edge being taken allows us to lower bound the overall approximation factor. This is expressed in the following lemma.
\begin{lemma}\label{lemma:graph-algorithm-bounded-by-h-delta}
    Let $W(S)$ be the sum of the weights of edges in $S$, and recall that $\accset$ is the set of edges accepted by Algorithm \ref{alg:graph}. Then,
    \begin{align*}
        \frac {\mathbb{E}[W(A)]} {W(OPT)} \geq \frac 1 m \sum_{t = m' + 1}^m h_{\frac 1 4}(t).
    \end{align*}
\end{lemma}
\begin{proof}
    For any element $e = (u, v)$ of $OPT$, we note as stated above that if it is presented at step $t$, then it is guaranteed to be in $\optmatch_i$. It follows that it will be taken iff when we randomly order $u, v$ as $a, b$, $seen\_outgoing(a) = False$ and $taken\_outgoing(b) = False$. Letting $S$ be the set of edges presented prior to step $t$, the probability that said conditions are satisfied conditioned on $S$ but not its order is $h_{u, v, S}(t)$. We know from Lemma \ref{lemma:graph-algorithm-h} that $h_{u, v, S}(t) \geq h_{\frac 1 4}(t)$ for any $S$. Therefore, given that $e$ is presented at step $t$, the probability that it is taken is at least $h_{\frac 1 4}(t)$.

    It then follows that the overall probability that $e$ is taken is the sum over steps after $m'$ of the probability that $e$ appears at that step and is taken. As the probability of $e$ being presented at a particular step is $\frac 1 m$, this gives the following:
    \begin{align*}
        P[e\text{ is taken}] \geq \sum_{t = m' + 1}^m \frac 1 m h_{\frac 1 4}(t).
    \end{align*}
    The expected ratio of the algorithm's output's weight to the optimal weight is at least the probability of any individual element of $OPT$ being taken, meaning that the desired conclusion follows.
\end{proof}

It now only remains to give a concrete lower bound for $\frac 1 m\sum_{t = m' + 1}^m h_{\frac 1 4}(t)$. We proceed to derive such a bound in the following lemmas.

\begin{lemma}\label{lemma:h-delta-bound}
    For any $\delta \geq 0$ and $i \geq m'$,
    \begin{align*}
        h_\delta(i) \geq \frac {m'(m' - 1)} {i^2} \left[1 + \delta \left(\ln \frac {i + 1} {m' + 1} - 1 + \frac {m'} i\right)\right].
    \end{align*}
\end{lemma}
\begin{proof}
    First recall that the definition of $h_\delta(i)$ is $(1 - \frac 2 i)h_\delta(i - 1) + l(i) = \frac {i - 2} i h_\delta(i - 1) + l(i)$ with the nonrecursive terms contained in $l(i)$. If we expand this definition out, the coefficient by which $h_\delta(j)$ will be multiplied is $\frac {i - 2} i \times \dotsb \times \frac {(j + 1) - 2} {j + 1} = \frac {j(j - 1)} {i(i - 1)}$. It follows by expanding the recursive definition fully that
    \begin{align*}
        h_\delta(i)
        &= \frac {m'(m' - 1)} {i(i - 1)} h_\delta(m') + \sum_{j = m' + 1}^i \frac {j(j - 1)} {i(i - 1)} l(i) &\EqComment{Expanding definition} \\
        &= \frac {m'(m' - 1)} {i(i - 1)}h_\delta(m') + \sum_{j = m' + 1}^i \frac {j(j - 1)} {i(i - 1)} \cdot \frac \delta j \cdot \frac {m'(m' - 1)} {(j - 1)(j - 2)}\left[1 - \frac {m'} {j - 1}\right] &\EqComment{Definition of $l(i)$}\\
        &\geq \frac {m'(m' - 1)} {i^2}h_\delta(m') + \sum_{j = m' + 1}^i \frac {\delta m'(m' - 1)} {i^2j}  \left[1 - \frac {m'} {j - 1}\right] &\EqComment{$j \geq j - 2$, $i \geq i - 1$}\\
        &= \frac {m'(m' - 1)} {i^2} \left[1 + \delta \sum_{j = m' + 1}^i \left(\frac 1 j - \frac {m'} {j(j - 1)}\right)\right]. &\EqComment{Rearranging}
    \end{align*}
    We can lower bound the sum using integrals. First, see that $\sum_{j = m' + 1}^i \frac 1 j \geq \int_{m' + 1}^{i + 1} \frac {dx} x = \ln \frac {i + 1} {m' + 1}$. Second, see that $\sum_{j = m' + 1}^i \frac {m'} {j(j - 1)} \leq \sum_{j = m' + 1}^i \frac {m'} {j^2} \leq m'\int_{m'}^i \frac {dx} {x^2} = m'(\frac 1 {m'} - \frac 1 i) = 1 - \frac {m'} i$. We can therefore write
    $
        h_\delta(i) \geq \frac {m'(m' - 1)} {i^2} \left[1 + \delta \left(\ln \frac {i + 1} {m' + 1} - 1 + \frac {m'} i\right)\right]$
    as desired.
    
\end{proof}
\begin{lemma}\label{lemma:h-delta-sum-bound}
    For any $\delta \geq 0$, $\sum_{i = m' + 1}^m h_\delta(i)$ is at least $m'(m' - 1)$ multipled by
    \begin{align*}
        \left(1 - \delta\left(\ln(m' + 1) + 1\right)\right)\left(\frac 1 {m'} - \frac 1 {m}\right) + \delta \frac {\ln\left(m' + 1\right) + 1} {m'} - \delta\frac {\ln(m + 1) + 1} m + \frac {\delta m'} {2(m' + 1)^2} - \frac {\delta m'} {2(m + 1)^2}.
    \end{align*}
\end{lemma}
\begin{proof}
    We first apply Lemma \ref{lemma:h-delta-bound} and separate out terms in the resulting sum:
    \begin{align*}
        \sum_{i = m' + 1}^m h_\delta(i) &\geq \sum_{i = m' + 1}^m \frac {m'(m' - 1)} {i^2} \left[1 + \delta \left(\ln \frac {i + 1} {m' + 1} - 1 + \frac {m'} i\right)\right]\\
        &= m'(m' - 1)\sum_{i = m' + 1}^m \left[\frac {1 - \delta(\ln(m' + 1) + 1)} {i^2} + \frac {\delta\ln(i + 1)} {i^2} + \frac {\delta m'} {i^3}\right].
    \end{align*}
    We can now lower bound each of these sums using integrals.
    \begin{itemize}
        \item Firstly, $\sum_{i = m' + 1}^m \frac {1 - \delta(\ln(m' + 1) + 1} {i^2} \leq (1 - \delta(\ln(m' + 1) + 1))\int_{m'}^m \frac {dx} {x^2} = (1 - \delta(\ln(m' + 1) + 1))(\frac 1 {m'} - \frac 1 {m})$.
        \item Secondly, $\sum_{i = m' + 1}^m \frac {\ln(i + 1)} {i^2} \geq \sum_{i = m' + 1}^m \frac {\delta\ln i} {i^2}\delta \int_{m' + 1} \geq \delta \int_{m' + 1}^{m + 1} \frac {\ln x} {x^2} dx = \delta\frac {\ln(m' + 1) + 1} {m'} - \delta\frac {\ln(m + 1) + 1} m$.
        \item Finally, $\sum_{i = m' + 1}^m \frac {\delta m'} {i^3} \geq \delta m'\int_{m' + 1}^{m + 1} \frac {dx} {x^3} = 
        \frac {\delta m'} 2(\frac 1 {(m' + 1)^2} - \frac 1 {(m + 1)^2})$.
    \end{itemize}
    We thus have as desired that\\$\sum_{i = m' + 1}^m h_\delta(i) \geq m'(m' - 1)\left[\left(1 - \delta\left(\ln(m' + 1) + 1\right)\right)\left(\frac 1 {m'} - \frac 1 {m}\right) + \delta \frac {\ln\left(m' + 1\right) + 1} {m'} - \delta\frac {\ln(m + 1) + 1} m + \frac {\delta m'} {2(m' + 1)^2} - \frac {\delta m'} {2(m + 1)^2}\right]$.
\end{proof}
\begin{lemma}\label{lemma:h-delta-limit}
    If we let $m' = \lfloor \alpha m \rfloor$, then
    \begin{align*}
        \lim_{m \to \infty} \frac 1 m \sum_{i = m' + 1}^m h_\delta(i) \geq \alpha - \alpha^2 + \delta \alpha^2 \ln \alpha + \frac \delta 2 \left(\alpha - \alpha^3\right).
    \end{align*}
\end{lemma}
\begin{proof}
    The proof proceeds by applying the lower bound from Lemma \ref{lemma:h-delta-sum-bound}, then applying straightforward algebraic manipulations. First note that as we are taking the limit, we can ignore the floor in $\lfloor \alpha m \rfloor$ and further replace all terms of the form $m' - 1$ and $m' + 1$ by $m' = \alpha m$ and all terms of the form $m - 1$ and $m + 1$ by $m$, because the resulting differences go to $0$ as $m$ goes to infinity. Therefore,
    \begin{align*}
        \lim_{m \to \infty} \sum_{i = m' + 1}^m \frac {h_\delta(i)} m
        &\geq \lim_{m \to \infty} \frac {\left(\alpha m\right)^2} m \left[(1 - \delta(\ln (\alpha m) + 1))\left(\frac 1 {\alpha m} - \frac 1 m\right) + \delta\frac {\ln (\alpha m) + 1} {\alpha m} - \delta\frac {\ln m + 1} m + \frac {\delta \alpha m} {2(\alpha m)^2} - \frac {\delta \alpha m} {2m^2}\right]\\
        &= \lim_{m \to \infty} \left[(1 - \delta(\ln (\alpha m) + 1))(\alpha - \alpha^2) + \delta \alpha \ln(\alpha m) + \delta \alpha - \delta \alpha^2\ln m - \delta \alpha^2 + \frac {\delta \alpha} 2 - \frac {\delta \alpha^3} 2\right]\\
        &= \lim_{m \to \infty} \left[\alpha - \alpha^2 + \delta \alpha^2 \ln \alpha + \frac \delta 2 \left(\alpha - \alpha^3\right)\right]\\
        &= \alpha - \alpha^2 + \delta \alpha^2 \ln \alpha + \frac \delta 2 \left(\alpha - \alpha^3\right).
    \end{align*}
    with the penultimate step involving multiple uses of $\ln(\alpha m) = \ln \alpha + \ln m$ and cancellations.
\end{proof}
We now complete the analysis by proving Theorem \ref{thm:graph-algorithm}.
\begin{proof}[Proof of Theorem \ref{thm:graph-algorithm}]
    We first prove that if $m$, the number of edges, is sufficiently large -- more precisely, if $m \geq M_n$ where $M_n$ may depend on $n$ -- then Algorithm \ref{alg:graph} is $3.95$-competitive for graphic matroid secretary with an appropriate choice of $m'$. To see this, first note by Lemma \ref{lemma:graph-algorithm-independent} that the set of edges taken by Algorithm \ref{alg:simple-graph}. Now note that if we let $\alpha = 0.4914$, then the lower bound for $\lim_{m \to \infty} \frac 1 m \sum_{i = m' + 1}^m h_{\frac 1 4}(i)$ provided by Lemma \ref{lemma:h-delta-limit} is greater than $0.2536$. It follows that there exists some $M$ such that for all $m \geq M$, $\lim_{m \to \infty} \frac 1 m \sum_{i = m' + 1}^m h_{\frac 1 4}(i) \geq 0.2536$. We then note by Lemma \ref{lemma:graph-algorithm-bounded-by-h-delta} that $\frac {\mathbb{E}[W(A)]} {W(OPT)}$, i.e. the inverse of the competitive ratio, for Algorithm \ref{alg:graph} is lower bounded by $\lim_{m \to \infty} \frac 1 m \sum_{i = m' + 1}^m h_{\frac 1 4}(i)$, from which it follows that if $m \geq M$, then Algorithm \ref{alg:graph} is $\frac 1 {0.2536} < 3.95$-competitive when we take $m' = \lfloor \alpha m \rfloor$.

    Additionally note (as previously mentioned in a footnote) that in order for Algorithm \ref{alg:graph} to successfully execute, we need that at every step after step $m'$, there exist enough edges in order to assign an outgoing edge to each node. This is equivalent to requiring that $m' + 1 \geq n$, for which it is sufficient to demand that $m \geq \frac n {\alpha}$. We can therefore set $M_n = \max(M, \frac n \alpha)$ where $M$ is as defined before to complete the proof of this step.

    We now prove the full theorem. To do this, when $m < M_n$, we simulate the execution of Algorithm \ref{alg:graph} on a graph with additional edges by mixing in dummy edges. Specifically, let $l = M_n - m$. Then we first construct a "bag" containing $m$ 1s and $l$ 0s, and choose permute the contents of this bag uniformly at random. We then execute Algorithm \ref{alg:graph} as if on a graph with $M_n$ edges as we iterate through this permutation. For each 1 that we see, we present Algorithm \ref{alg:graph} with the next real edge, while for each 0 that we see, we present Algorithm \ref{alg:graph} with a dummy edge. The endpoints of the dummy edge do not matter as we can simply modify Algorithm \ref{alg:graph} to never include the dummy edges in $\optmatch_t$, which will mean that the if statement whose body actually references the endpoints of the edge will never execute.

    As the optimum in the modified graph created by the inclusion of the dummy edges is the same as the optimum of the original graph, and the algorithm will never take a dummy edge as one will never be in $\optmatch_t$, the competitive ratio attained by the algorithm on the modified graph will carry over to the original graph, and so the algorithm modified in this way will be $3.95$-competitive. This description therefore gives a $3.95$-competitive algorithm for graphic matroid secretary as desired.
\end{proof}

\subsection{Improvement for simple graphs}
\label{section:simple}

\paragraph{Overview.}

When the graph $G$ is restricted to be a simple graph, meaning that between any two nodes $u, v$ there exists at most one edge, we can obtain an improved competitive ratio by modifying \ref{alg:graph}. Specifically, rather than randomly choosing which endpoint $a$ of each edge $e$ to which to apply the stronger condition that no outgoing edge from $a$ had been seen, we always choose it to be the endpoint $v$ such that $e$ is not the outgoing edge from $v$. Pseudocode for this modified algorithm is provided in Figure~\ref{alg:simple-graph}

\begin{algorithm}[!ht]
  Let $E'$ be the first $m'$ edges \hfill \Comment{set of observed edges}
  $\accset \gets \emptyset$ \hfill\Comment{set of accepted edges}
  
  $\forall v\in V$: $seen\_outgoing(v)\gets False$\\
  $\forall v\in V$: $taken\_outgoing(v)\gets False$\\
  \For{$t \in \cbr{m' + 1, \dots, m}$} {
    Let $e_t$ be the edge arriving in time $t$\\
    Add $e_t$ to $E'$ \\
    $\optmatch_t \gets $ maximum weighted directed
    forest on $G[E']$\\
    $\forall v \in V: outgoing(v) \gets $ the edge directed away from $v$ in $\optmatch_t$ if it exists\\
    $\forall v \in V$ such that $outgoing(v) = null : outgoing(v) \gets $ a unique edge in $E'$ not in $\optmatch_t$\\
    \If{$e_{t} \in \optmatch_{t}$} { 
    $e_{t} \gets (u, v)$ \hfill
    \Comment{where $e_t$ is directed from $u$ to $v$ in $\optmatch_{t}$
    }
    
    \If{$seen\_outgoing(v) = False$ and $taken\_outgoing(u) = False$}
    {
      $taken\_outgoing(u) \gets True$ \\
      add $e_{t}$ to $\accset$
    }
    
  }
  \If{$outgoing(w) = e_t$ for some $w \in V$} {
    $seen\_outgoing(w) \gets True$
  }
  }
  \caption{Algorithm for graphic matroid secretary on simple graphs.}
  \label{alg:simple-graph}
\end{algorithm}

To understand how this improvement is possible, recall that the motivation for the random choice of $(a, b)$ in Algorithm \ref{alg:graph} was that the stronger condition applied to node $a$ in that algorithm yields some improvement in the competitive ratio by allowing us to consider the possibility that a previous edge $e'$ that was outgoing edge from an endpoint of the current edge $e$ was not actually taken. This stronger condition provides a concrete lower bound for this possibility. However, we must simultaneously ensure that previous edges also would not prevent $e$ from being taken, and if the endpoints of $e'$ are the same as those of $e$, then $e'$ not being taken could only mean that $e$ would not be taken either (see Figure \ref{fig:four_cases} which demonstrated a "worst-case" example where only one of the four random choices led to an improvement),

In simple graphs, it is not possible for $e'$ to have the same endpoints as $e$, meaning that this random choice is unnecessary. We can therefore fix the vertex to which we apply the stronger condition to be the endpoint for which $e'$ is not the outgoing edge, as this endpoint will then be the endpoint not in common with $e$, ensuring that the possibility of $e'$ not being taken for a reason that does not also prevent $e$ from being taken always exists.

\paragraph{Analysis.}

We aim to prove the following theorem:
\begin{theorem}\label{thm:simple-graph-algorithm}
    There exists a $3.77$-competitive algorithm for graphic matroid secretary when the graph is restricted to be a simple graph.
\end{theorem}
We first note the following lemma demonstrating the validity of Algorithm \ref{alg:simple-graph}. We omit the proof as it is identical to the proof of Lemma \ref{lemma:graph-algorithm-independent}.
\begin{lemma}\label{lemma:simple-graph-algorithm-independent}
    The set $\accset$ of accepted edges in Algorithm \ref{alg:simple-graph} is always an independent set in the graphic matroid.
\end{lemma}
The remainder of the analysis essentially consists of modifying the analysis of Algorithm \ref{alg:graph} to first show that the modification contained in Algorithm \ref{alg:simple-graph} allows us to use $\delta = 1$ instead of $\delta = \frac 1 4$, then apply the lemmas shown previously bounding $h_\delta$ to derive the new competitive ratio.

We redefine $f$ and $g$ below; note that these definitions are in fact identical to their definitions in the previous analysis.
\begin{definition}
    Define $f_{u, S}(i)$ to be the probability that after step $i$, $seen\_outgoing(u) = False$, given that $S$ is the set of edges that are presented in steps $1, \ldots, i$.
\end{definition}
\begin{definition}
    Define $g_{u, v, S}(i)$ to be the probability that after step $i$, both $seen\_outgoing(u) = False$ and $seen\_outgoing(v) = False$, given that $S$ is the set of edges that are presented in steps $1, \ldots, i$.
\end{definition}
The definition below is analogous to the definition of $h_{u, v, S}$ from before; in particular, the condition used to define $h^{\text{simple}}$ is the same condition used to check whether an edge between $u$ and $v$ can be taken in Algorithm \ref{alg:simple-graph}.
\begin{definition}
    Define $h^{\text{simple}}_{u, v, S}(i)$ to be the probability that after step $i$, $seen\_outgoing(v) = False$ and $taken\_outgoing(u) = False$, given that $S$ is the set of edges that are presented in steps $1, \ldots, i$.
\end{definition}
Additionally recall the definition of $h_\delta$:
\begin{definition}
    Define $h_\delta(i)$ recursively by $h_a(m') = 1$ and $h_a(i) = (1 - \frac 2 i)h_a(i - 1) + \frac \delta i \cdot \frac {m'(m' - 1)} {(i - 1)(i - 2)} \left[1 - \frac {m'} {i - 1}\right].$
\end{definition}
We now prove the following lemma analogous to Lemma \ref{lemma:graph-algorithm-h}.
\begin{lemma}\label{lemma:simple-graph-algorithm-h}
    For all vertices $u, v$, all $i$ from $m'$ to $n$, and all $S$ with $|S| = i$, $h^{\text{simple}}_{u, v, S}(i) \geq h_1(i).$
\end{lemma}
\begin{proof}
    The proof is by induction. We have $h^{\text{simple}}_{u, v, S}(m') = 1 = h_1(m')$. Then, for any $i > m'$, we have that with probability $1 - \frac 2 i$ the edge that is presented at step $i$ is not directed from either of $u, v$; given this, the probability is a weighted average of terms $h^{\text{simple}}_{u, v, S'}(i - 1)$, which we can lower bound by $h_1(i - 1)$. This contributes $(1 - \frac 2 i)h_1(i - 1)$ to our lower bound.

    In the case that an edge $e'$ directed away from $u$ is selected, note that as the input graph is a simple graph, $e'$ must connect to a third vertex $w$. There is a $\frac 1 i$ probability that $e'$ is presented at step $i$. Then, the probability that $seen\_outgoing(w) = True$ by step $i$, and therefore $e'$ is taken, is $(1 - f_{w, S'}(i - 1))$.
    
    We therefore have an additional probability of $\frac 1 i \cdot (1 - f_{w, S'}(i - 1))$ that the edge presented at step $i$ does not violate the conditions defining $h^{\text{simple}}$. We additionally require that edges prior to step $i$ did not violate said conditions. Note that we are conditioned on the fact that $seen\_outgoing(w) = True$ by step $i$; we therefore impose the stronger condition that $seen\_outgoing(u) = seen\_outgoing(v) = False$ by step $i$, as the likelihood of this is only increased by the fact that $selected(w) \neq null$. We can thus use $g_{u, v, S''}(i - 1)$ as a lower bound for said probability.

    All of this combines to give us
    \begin{align*}
        h^{\text{simple}}_{u, v, S}(i)
        &\geq \left(1 - \frac 2 i\right)h_1(i - 1) + \frac 1 i  \cdot (1 - f_{w, S'}(i - 1)) \cdot g_{u, v, S''}(i - 1)\\
        &= \left(1 - \frac 2 i\right)h_1(i - 1) + \frac 1 i \cdot \left[1 - \frac {m'} {i - 1}\right] \cdot \frac {m'(m' - 1)} {(i - 1)(i - 2)} &\EqComment{Lemmas \ref{lemma:graph-algorithm-f} and \ref{lemma:graph-algorithm-g}}\\
        &= h_1(i). &\EqComment{Definition \ref{def:h-delta}}
    \end{align*}
\end{proof}

As with the previous analysis, the conditions defining $h^{\text{simple}}_{u, v, S}(i)$ are identical to the conditions in Algorithm \ref{alg:simple-graph} for taking an edge given that it is in $\optmatch_i$. As any edge in the overall optimum is in $\optmatch_i$, lower bounding the probability of such an edge being taken allows us to lower bound the overall approximation factor. This is expressed in the following lemma, whose proof we omit as it is analogous to the proof of Lemma \ref{lemma:graph-algorithm-bounded-by-h-delta}.
\begin{lemma}\label{lemma:simple-graph-algorithm-bounded-by-h-delta}
    Let $W(S)$ be the sum of the weights of edges in $S$, and recall that $\accset$ is the set of edges accepted by Algorithm \ref{alg:simple-graph}. Then,
    \begin{align*}
        \frac {\mathbb{E}[W(A)]} {W(OPT)} \geq \frac 1 m \sum_{t = m' + 1}^m h_1(t).
    \end{align*}
\end{lemma}

We now complete the analysis by proving Theorem \ref{thm:simple-graph-algorithm}.

\begin{proof}[Proof of Theorem \ref{thm:simple-graph-algorithm}]

    As in the proof of Theorem \ref{thm:graph-algorithm}, we first prove that if $m$, the number of edges, is sufficiently large -- more precisely, if $m \geq M_n$ where $M_n$ may depend on $n$ -- then Algorithm \ref{alg:simple-graph} is $3.77$-competitive for graphic matroid secretary with an appropriate choice of $m'$. To see this, first note by Lemma \ref{lemma:simple-graph-algorithm-independent} that the set of edges taken by Algorithm \ref{alg:simple-graph}. Now note that if we let $\alpha = 0.4642$, then the lower bound for $\lim_{m \to \infty} \frac 1 m \sum_{i = m' + 1}^m h_1(i)$ provided by Lemma \ref{lemma:h-delta-limit} is greater than $0.2654$. It follows that there exists some $M$ such that for all $m \geq M$, $\lim_{m \to \infty} \frac 1 m \sum_{i = m' + 1}^m h_1(i) \geq 0.2654$. We then note by Lemma \ref{lemma:simple-graph-algorithm-bounded-by-h-delta} that $\frac {\mathbb{E}[W(A)]} {W(OPT)}$, which is the inverse of the competitive ratio, for Algorithm \ref{alg:simple-graph} is lower bounded by $\lim_{m \to \infty} \frac 1 m \sum_{i = m' + 1}^m h_1(i)$, from which it follows that if $m \geq M$, then Algorithm \ref{alg:simple-graph} is $\frac 1 {0.2654} < 3.77$-competitive when we take $m' = \lfloor \alpha m \rfloor$.
    
    Additionally note, as for Algorithm \ref{alg:graph}, that in order for Algorithm \ref{alg:graph} to successfully execute, we need that at every step after step $m'$, there exist enough edges in order to assign an outgoing edge to each node. This is equivalent to requiring that $m' + 1 \geq n$, for which it is sufficient to demand that $m \geq \frac n {\alpha}$. We can therefore set $M_n = \max(M, \frac n \alpha)$ where $M$ is as defined before to complete the proof of this step.

    We now prove the full theorem. To do this, when $m < M_n$, we simulate the execution of Algorithm \ref{alg:simple-graph} on a graph with additional edges by mixing in dummy edges. This modified algorithm then achieves the desired competitive ratio.
    
    As this modification is fully described in the proof of Theorem \ref{thm:graph-algorithm}, we do not repeat said description. We instead note that there may appear to be an issue with the strategy of introducing dummy edges in a simple graph, as the total number of edges is limited to $\binom n 2$. However, as described previously, the endpoints of the introduced dummy edges are irrelevant, because they are excluded from $\optmatch_t$, and so they do not need to have actual endpoints at all, meaning that we can add an arbitrarily large amount of them without invalidating the execution of the algorithm in any way. This can alternatively be interpreted as introducing $2(M_n - m)$ dummy nodes which we connect with $M_n - m$ dummy edges of weight $0$.

\end{proof}

\section{Impossibility result for high girth graphs: proof of Theorem~\ref{lowerboundthm}}

In this section we will prove our hardness result for high girth graphs. We first prove that the weighted single secretary problem is hard even if instance is drawn from a fixed distribution known to the algorithm. The proof of existence of this distribution uses Finite Ramsey Theorem and the Strong Duality Theorem from zero-sum games. We then show how to reduce this problem to the weighted graphic matroid secretary problem on graphs with large girth using a construction based on Ramanujan graphs. Specifically, we show the existence of bipartite graphs with certain properties by starting with a Ramanujan graph, then applying the probabilistic method to remove edges appropriately.
 We then show how to embed the hard instances of the weighted single secretary problem in these bipartite graphs in order to create a similarly hard instance of the weighted graphic secretary problem.
 
 We note that it is crucial that we utilize a distribution for which we have hardness even when the distribution is known up front, as opposed to simpler deterministic instances which are hard when the algorithm does not already know the instance. This is because, roughly speaking, we are able to present multiple instances drawn from such a distribution to the algorithm without the information the algorithm learns about the distribution from earlier instances aiding it when handling later instances -- the algorithm already completely knows the distribution.

\subsection{Hard inputs for single secretary}

Given an algorithm $a \in \mathcal{A}$ and an instance $I$ of the single secretary problem, $a(I)$ denotes the set of elements output by $a$ on instance $I$. Let $opt(I)$ be the optimal solution to the problem on the instance $I$. Given any subset $X$ of the elements, we denote by $w(X)$ the sum of the weights of the elements in $X$.

We say that a (randomized) algorithm $a \in \mathcal{A}$ for the single secretary problem is $\alpha$-probability competitive over a finite set of weights $W$ if for any instance of the problem where weights are chosen from $W$, the probability that $a$ selects the maximum element in the instance is at least $1/\alpha$,
where the probability is over the internal randomness of algorithm $a$ and over the random choice of the elements' order.

We will build on the proof of Correa et al.~\cite{CorreaDFS21,CorreaDFS19} to prove the following theorem. The main technical difference between that proof and ours is that we work with single secretary rather than single choice prophet inequality algorithms, and we require that the set of weights $W$ is finite, whereas the authors of \cite{CorreaDFS21} construct an infinite set.

\begin{theorem}
\label{thm:exists_W}
  For any $\varepsilon > 0$ and for any $\rho > 1$, there exists an integer $n$ and a finite set of numbers $W \subset \{1, \rho, \rho^2, \ldots\}$ over which no (randomized) algorithm for the single secretary problem with $n$ can choose the maximum with probability $\frac 1 e + \varepsilon$.
\end{theorem}

Given a set $X$ and a positive integer $k \in \mathbb N$, let $\binom X k$ denote the collection of all $k$-element subsets of $X$.
The following finite version of his famous theorem was proved by Ramsey in his original paper \cite{Ramsey1930} as Theorem B.

\begin{theorem}[Finite Ramsey Theorem, \cite{Ramsey1930}]\label{thm:FiniteRamsey}
  Let $W$ be any finite set and and $c,n \in \mathbb N, n \leq |W|$, and let us arbitrarily color all elements of $\binom W n$ using $c$ different colors; such a coloring is any mapping $\phi : \binom W n\rightarrow \{1,2,\ldots,c\}$. For any given numbers $c, n, T \in \mathbb N$, there exists a $k \in \mathbb N$ such that if $|W| \ge k$, then there is a subset $W' \subseteq W$ of size $|W'| \ge T$, such that all subsets  $\binom {W'} n $ are monochromatic, i.e., colored by the same color.
\end{theorem}

We will formalize a notion of a (randomized) algorithm for the single secretary problem as follows. Let the set of all randomized algorithms for the single secretary problem be denoted by $\mathcal{A}$. Given an algorithm $a \in \mathcal{A}$, we denote by $v_1,v_2,\ldots,v_n$ the sequence of values that the algorithm sees in subsequent times, i.e., $a$ sees vale $v_i$ at time $i$. Then let us denote by $a_i(v_1,\ldots,v_i) \in [0,1]$ the probability that algorithm $a$ stops at time $i$ and accepts value $v_i$, conditioned on having seen the values $v_1,\ldots,v_i$ and not having stopped on any of the previous values $v_1,\ldots,v_{i-1}$. 

We will define first a notion of an order-oblivious algorithm following \cite{CorreaDFS21}. 

\begin{definition}
An algorithm $a \in \mathcal{A}$ is called {\em order-oblivious} if for each $j \in [n]$, and all pair-wise disjoint values $v_1,\ldots,v_j \in \reals_{\geq 0}$, and all permutations $\pi \in \mathcal{S}_{j-1}$, we have  
$a_j(v_1,\ldots,v_j) = a_j(v_{\pi(1)},\ldots,v_{\pi(j-1)}, v_j)$. 
\end{definition}

In words, the decision of an order-oblivious algorithm $a$ at any time $j$ depends only on the set of previous elements and the current element at time $j$, i.e., it does not depend on the order of the previous elements up to time $j-1$. We first note the following lemma; this lemma is proven as the slightly stronger Lemma 2 in \cite{CorreaDFS21}.

\begin{lemma}\label{lemma:order-oblivious}
  Suppose that there is a (randomized) algorithm for the single secretary problem that is $\alpha$-probability competitive over a set of weights $W$. Then there is another algorithm for this problem that is also $\alpha$-probability competitive over $W$, and is additionally order-oblivious.
\end{lemma}

We then need a notion of $\varepsilon$-value-oblivious algorithm from \cite{CorreaDFS21}. 

\begin{definition}
Given any $\varepsilon > 0$ and any set $W \subseteq \R_{+}$, algorithm $a \in \mathcal{A}$ is called {\em $\varepsilon$-value-oblivious} on set $W$ if, for all $i \in [n]$, there exists $q_i \in [0,1]$ such that, for all pairwise distinct $v_1,\ldots,v_i \in W$ with $v_i > \max\{v_1,\ldots,v_{i-1}\}$, we have that $a_i(v_1,\ldots,v_i) \in [q_i - \varepsilon, q_i + \varepsilon)$. 
\end{definition}

Such an algorithm's decision of whether to stop at any time $i$ does not depend too much on the previously seen values $v_1,\ldots,v_{i-1}$; note also that the notion of value-obliviousness is importantly connected to the set of weights $W$. 

We will now need the following lemma, which corresponds to Lemma 1 in \cite{CorreaDFS21}. For the purpose of proving this lemma, we will also need a notion of $(\varepsilon,i)$-value-obliviousness. Given any $\varepsilon > 0$, $i \in [n]$, and any set $W \subseteq \R_{+}$, algorithm $a \in \mathcal{A}$ is called {\em $(\varepsilon,i)$-value-oblivious} on set $W$ if there exists $q \in [0,1]$ such that, for all pairwise distinct $v_1,\ldots,v_i \in W$ with $v_i > \max\{v_1,\ldots,v_{i-1}\}$, we have that $a_i(v_1,\ldots,v_i) \in [q - \varepsilon, q + \varepsilon)$.

Note that an algorithm is $\varepsilon$-value-oblivious iff it is $(\varepsilon,i)$-value-oblivious for all $i \in [n]$.

\begin{lemma}\label{lemma:epsilon-value-oblivious}
  Let $\rho > 1$ be a real number, and let us denote $\Psi = \{1, \rho, \rho^2, \ldots\}$.
  For any $\varepsilon > 0$ and $n \in \nats$ there exists a finite set $W \subset \Psi$ with the following property.
  For any algorithm $a \in \mathcal{A}$ for the single secretary problem with $n$ items that is $\alpha$-probability competitive,
  there exists another algorithm $a' \in \mathcal{A}$ for the same problem that is also $\alpha$-probability competitive and a finite set
  $W' \subseteq W$, such that algorithm $a'$ is $\varepsilon$-value-oblivious over $W'$ and $|W'| = h(n) \geq n$, for any function $h : \nats \rightarrow \nats$.
\end{lemma}

\begin{proof}
Let $W$ be a subset of $\Psi$ to be chosen later. Let $a \in \mathcal{A}$ be any algorithm for the single secretary problem with $n$ items that is $\alpha$-probability competitive. By Lemma \ref{lemma:order-oblivious}, we can assume that algorithm $a$ is order-oblivious.
 
We now fix any $\varepsilon > 0$ and will show by induction on $j \in [n]$ the following claim:

\noindent
{\bf (*)} there exists a finite set $W_j \subset W$ such that, for all $i \in [j]$, algorithm $a$ is $(\varepsilon,i)$-value-oblivious on $W_j$. 

For $j=n$, this will imply that $a$ is $(\varepsilon,j)$-value-oblivious on $W_n$ for all $j \in [n]$; hence $a$ is $\varepsilon$-value-oblivious on $W_n$.

We will intend to apply the Finite Ramsey's Theorem \ref{thm:FiniteRamsey} in the following way. Let $c, \ell \in \nats$. Suppose that we are given any finite set $W_{\ell-1}$ of size $|W_{\ell-1}| = w_{\ell-1}$, whose subsets ${W_{\ell-1} \choose \ell}$ we arbitrarily color using $c$ colors. Let $w_{\ell} \in \nats$ be any integer. By the Finite Ramsey's Theorem, there exists $w_{\ell-1} \in \nats$, such that if $|W_{\ell-1}| \geq w_{\ell-1}$ (let us take exactly $|W_{\ell-1}| = w_{\ell-1}$), then there is a subset $W_{\ell} \subseteq W_{\ell-1}$ of size $|W_{\ell}| \geq w_{\ell}$ (let us take exactly $|W_{\ell}| = w_{\ell}$), such that all subsets ${W_{\ell} \choose \ell}$ are monochromatic.

Given $c, \ell \in \nats$ and $w_{\ell} \in \nats$, let us denote this Ramsey number $w_{\ell-1}$ as 
$w_{\ell-1} = \mathcal{R}(w_{\ell}, \ell, c)$. Let $c = \lceil 1/2\varepsilon \rceil$. We define consecutive Ramsey numbers which we will use in the proof:
$$
 w_n = h(n), \,\,\, w_{n-1} = \mathcal{R}(w_{n}, n, c), \,\,\, w_{n-2} = \mathcal{R}(w_{n-1}, n-1, c), \,\, \ldots, \,\,
 w_{1} = \mathcal{R}(w_{2}, 2, c), \,\,\, w_{0} = \mathcal{R}(w_{1}, 1, c) \, .
$$ This implies that $w_0$ is defined as follows:
$$
   w_0 = \mathcal{R}(w_{1}, 1, c) = \mathcal{R}(\mathcal{R}(w_{2}, 2, c), 1, c) = \ldots = \underbrace{\mathcal{R}(\mathcal{R}( \ldots (\mathcal{R}}_{n}(\overbrace{h(n)}^{w_n}, \underbrace{n, c), n-1,c), \ldots),1,c)}_{n} \, .
$$

We will now proceed with the inductive proof to show claim {\bf (*)} by induction on  $j \in [n]$. Let the initial set $W_0 \subset \Psi$ be {\em any} subset of consecutive positive pair-wise distinct numbers from set $\Psi$ of size $|W_0| = w_0$. We then now choose $W = W_0$, noting that the definition of $W_0$ did not depend on the algorithm $a$. The set $W_0$ thus clearly satisfies the claim {\bf (*)} for $j=0$. We will show the claim for $j=\ell > 0$, assuming that it holds for $j < \ell$. Observe that it suffices to find a finite set $W_{\ell} \subseteq W_{\ell-1}$ such that algorithm $a$ is $(\varepsilon,\ell)$-value-oblivious on $W_{\ell}$, because the induction hypothesis implies $(\varepsilon,i)$-value-obliviousness of algorithm $a$ on $W_{\ell}$ as a subset of $W_i$ for all $i \in [\ell-1]$.

To apply the Finite Ramsey's Theorem, we define first a specific coloring of ${W_{\ell-1} \choose \ell}$ with $c$ colors from the set $\{1,2,\ldots,c = \lceil 1/2\varepsilon \rceil\}$. Consider any set $\{v_1,\ldots,v_{\ell}\} \in {W_{\ell-1} \choose \ell}$ such that 
$v_{\ell} > \max\{v_1,\ldots,v_{\ell-1}\}$. Observe that there exists a unique number $u \in \{1,2,\ldots,\lceil 1/2\varepsilon \rceil\}$ such that $a_{\ell}(v_1,\ldots,v_{\ell}) \in [(2u-1) \cdot \varepsilon - \varepsilon, (2u-1) \cdot \varepsilon + \varepsilon)$. We color the $\ell$-element set $\{v_1,\ldots,v_{\ell}\}$ with color $u$.

By the Finite Ramsey's Theorem, there exists a finite subset $W_{\ell} \subseteq W_{\ell-1}$ of size $|W_{\ell}| = w_{\ell}$, such that all subsets ${W_{\ell} \choose \ell}$ are monochromatic; suppose that this monochromatic color is $u \in \{1,2,\ldots,c\}$. Let us now take $q = (2u-1) \cdot \varepsilon$ and consider any $\ell$ distinct values $v_1,\ldots,v_{\ell} \in W_{\ell}$ such that
 $v_{\ell} > \max\{v_1,\ldots,v_{\ell-1}\}$. Because set $\{v_1,\ldots,v_{\ell}\} \in {W_{\ell-1} \choose \ell}$ has color $u$, we have that $a_{\ell}(v_{\pi(1)},\ldots,v_{\pi(\ell-1)},v_{\ell}) \in [q - \varepsilon, q + \varepsilon)$ for some permutation $\pi \in \mathcal{S}_{\ell-1}$. But because algorithm $a$ is order-oblivious, it also holds that $a_{\ell}(v_{1},\ldots,v_{\ell-1},v_{\ell}) \in [q - \varepsilon, q + \varepsilon)$. Therefore, algorithm $a$ is $(\varepsilon,\ell)$-value-oblivious on $W_{\ell}$. This finished the proof of the inductive step.
 
To conclude the proof of the lemma, we now simply take $W' = W_n$.
\end{proof}

We now use the above lemma to prove the main theorem.

\begin{proof}[Proof of Theorem~\ref{thm:exists_W}]

Say that an algorithm $a \in \mathcal A$ is \emph{optimally value-oblivious} if, for all $i \in [n]$, there exists $q_i \in [0,1]$ such that, for all pairwise distinct $v_1,\ldots,v_i \in W$ with $v_i > \max\{v_1,\ldots,v_{i-1}\}$, we have that $a_i(v_1,\ldots,v_i) = q_i$, and for all for all pairwise distinct $v_1,\ldots,v_i \in W$ with $v_i < \max\{v_1,\ldots,v_{i-1}\}$,  we have that $a_i(v_1,\ldots,v_i) = 0$. An algorithm that it is optimally value-oblivious then acts in a way that depends only on the relative order of elements, and not their specific weights; it then follows from the optimal solution to the original secretary problem \cite{Ferguson1989who} that for any integer $n$, no optimally value-oblivious algorithm for the single secretary problem with $n$ items can have a probability of selecting the maximum element that is greater than $\frac 1 e + o(1)$. It follows that there exists some $n$ such that no optimally value-oblivious algorithm for the single secretary problem with $n$ items has a probability of choosing the maximum that is at least $\frac 1 e + \frac \varepsilon 2$. Note that as the algorithm's actions depend only on the relative order of elements and not on the specific weights, the bound of $\frac 1 e + \frac \varepsilon 2$ applies to \emph{any} instance of $n$ item single secretary where the items are distinct.

Therefore, fix such an $n$. Now, apply Lemma \ref{lemma:epsilon-value-oblivious} with parameters $\rho$, $\frac \varepsilon {2n}$, $n$, and $h(n') = n'$ to obtain a finite set $W$ satisfying the conditions of said lemma. We will now prove the theorem for $n, W$, (noting that Lemma \ref{lemma:epsilon-value-oblivious} gives us that $W \subseteq \{1, \rho, \rho^2, \ldots\}$) meaning that we would like to show that for any algorithm $a \in \mathcal A$ that it is $\alpha$-probability competitive, $\alpha$ satisfies $\frac 1 \alpha \leq \frac 1 e + \varepsilon$. To see this, first use the property of $W$ derived from Lemma \ref{lemma:epsilon-value-oblivious} that there exists an algorithm $a' \in \mathcal A$ that is also $\alpha$-probability competitive over $W$ and which is $\frac \varepsilon {2n}$-value-oblivious on some set $W'$ of size at least $n$.

Now define an algorithm $a'' \in \mathcal A$ as a relaxation of $a'$. Specifically, as $a'$ is $\frac \varepsilon {2n}$-value-oblivious over $W'$, there exist $q_i$ such that for all pairwise distinct $v_1, \ldots, v_i \in W'$ with $v_i > \max\{v_1, \ldots, v_{i - 1}\}$, $a'_i(v_1, \ldots, v_i) \in [q_i - \frac \varepsilon {2n}, q_i + \frac \varepsilon {2n})$. Therefore, fix any such $q_i$ and define $a''$ by $a''_i(v_1, \ldots, v_i) = q_i$ for all pairwise distinct $v_1, \ldots, v_i \in W'$ such that $v_i > \max\{v_1, \ldots, v_{i - 1}\}$; when $v_i < \max\{v_1, \ldots, v_{i - 1}\}$ we simply take $a''_i(v_1, \ldots, v_i) = a'_i(v_1, \ldots, v_i)$. Now note that if we defined $a'''$ alternately by instead saying that when $v_i < \max\{v_1, \ldots, v_{i - 1}\}$ we  take $a'''_i(v_1, \ldots, v_i) = 0$, then $a'''$ would always have at least as high of a probability of selecting the maximum as $a''$, as selecting an element known not to be the maximum is never useful. However, $a'''$ is then an optimally value-oblivious algorithm, meaning that on any instance with distinct weights its probability of choosing the maximum is less than $\frac 1 e + \frac \varepsilon 2$; the same is then true of $a''$.

Given this, consider alternately defining the execution of $a'$ and $a''$ as follows. For any instance $I$ of the single secretary problem with $n$ items over $W'$, we both randomly permute $I$ and also choose $n$ numbers $x_1, \ldots, x_n$ independently and uniformly at random from $[0, 1]$. Then, at step $i$, given that the elements seen are $v_1, \ldots, v_i$ in that order, each algorithm $b$, if it has not already stopped prior to step $i$, stops at step $i$ if and only if $x_i < b_i(v_1, \ldots, v_i)$. It is easy to see that the execution of each of $a', a''$ individually in this model is the same as their usual execution.

Observe that for any $i$ and pairwise distinct $v_1, \ldots, v_i \in W'$, the difference between $a'_i(v_1, \ldots, v_i)$ and $a''_i(v_1, \ldots, v_i)$ is at most $\frac \varepsilon {2n}$. It follows that in the above model of execution, the probability that the actions taken by $a'$ and $a''$ differ on step $i$ given that they have taken the same actions prior to step $i$ is at most $\frac \varepsilon {2n}$. Therefore, by the union bound, as there are $n$ steps, the probability that they ever differ is at most $n \cdot \frac \varepsilon {2n} = \frac \varepsilon 2$.

We can therefore roughly argue that because $a''$ takes the maximum with probability less than $\frac 1 e + \frac \varepsilon 2$ and $a', a''$ differ with probability at most $\frac \varepsilon 2$, $a'$ must take the maximum with probability less than $\frac 1 e + \varepsilon$. Formally, define $S$ to be the event in the above model of execution that the actions taken by $a'$ and $a''$ are identical, and define $T', T''$ to be the event that $a', a''$ respectively take the maximum element and $T''$. We have seen that $P(T') \geq \frac 1 \alpha$, $P(T'') < \frac 1 e + \frac \varepsilon 2$, and $P(S^C) \leq \frac \varepsilon 2$. Furthermore, by the definition of $S$, the events $T' \land S$ and $T'' \land S$ are the same. It follows that
\begin{align*}
    \frac 1 \alpha = P(T') &\leq P(T' \land S) + P\left(T' \land S^C\right)\\
    &= P(T'' \land S) + P\left(T' \land S^C\right)\\
    &\leq P(T'') + P\left(S^C\right)\\
    &< \left(\frac 1 e + \frac \varepsilon 2\right) + \frac \varepsilon 2 = \frac 1 e + \varepsilon.
\end{align*}
We therefore have that $\alpha$ satisfies $\frac 1 \alpha \leq \frac 1 e + \varepsilon$ as desired.
\end{proof}

\subsubsection{Minimax argument}

We will now use the Minimax principle and Theorem \ref{thm:exists_W} to prove the following.

\begin{theorem}\label{thm:exists_Distr_W}
  For any $\varepsilon > 0$ and for any $\rho > 1$, there exists an integer $n$, a finite set $W \subset \{1, \rho, \rho^2, \ldots\}$, and a finite distribution $D$ over all subsets of size $n$ from $W$ with the following property.
  If a random instance of the secretary problem is drawn from $D$, for any algorithm, the probability of choosing the maximum element from the instance random instance is at most $\frac 1 e + \varepsilon$, where the probability is over the randomness of the instance and the algorithm.
\end{theorem}
\begin{proof}
  Let $W$ and $n$ denote the values obtained from Theorem~\ref{thm:exists_W}.
  Consider the following finite game between
  two players which we refer to
  as the algorithm \enquote{designer} and 
  the \enquote{adversary}.
  The set of actions for the adversary is all finite sets of size $n$ from $W$.
  The set of actions for the designer is all deterministic algorithms for the secretary problem with instances that have size $n$ and are drawn from $W$.
  Given a deterministic algorithm $A$ chosen by the designer, and an instance $x$ chosen by the adversary, we define the utility of the designer
  to be the probability that $A$ chooses the maximum value in the instance $x$, where the probability here is over the random arrival order of the elements in $x$. The game is zero-sum and the adversary's goal is to minimize the designer's utility.

  We first claim that the set of actions for both the designer and the adversary is finite. For the adversary, this is clear since $n$ is finite. For the designer, observe that each deterministic algorithm can be uniquely described by a boolean function $f: (\cup_{i=1}^{n} W^{i}) \to \cbr{0, 1}$
  indicating whether or not it would accept the last element if presented with a tuple $w$ of size $i$ which has elements chosen from $W$. Since $n$ and $W$ are finite, the set of possible deterministic algorithms is finite as well.

  We next consider the mixed strategies in the game. 
  The set of mixed strategies for the algorithm designer is a the set of all randomized algorithms. This is because any randomized algorithm can be made deterministic by fixing the \enquote{seed} value used by the algorithm. 
  A mixed strategy for the adversary is a distribution over~the~instances.

  By the Von Neumann's minimax theorem~\cite{VonNeumann28}, if the players use mixed strategies, then the expected utility of the designer is the same regardless of which player moves first.
  Note that we can invoke the theorem as the game is finite.
  If the designer moves first, then the adversary can choose a distribution over instance after observing the randomized algorithm of the designer. By Theorem~\ref{thm:exists_W},
  for any such randomized algorithm, the adversary can choose an instance of size $n$ from $W$ for which the expected utility of the algorithm is at most $\frac 1 e + \varepsilon$. 
  If the adversary moves first and chooses some distribution $D$ over the instances, then the designer can choose its (randomized) algorithm after observing $D$. 
  Since the expected utility of the algorithm is the same in both cases, it follows that there exists a distribution $D$ over instances for which the designer cannot choose an algorithm which chooses the maximum with probability greater than $\frac 1 e + \varepsilon$, finishing the proof.
\end{proof}

\subsection{Embedding single secretary in simple graphs}

It is straightforward to reduce the single secretary problem to the graphical matroid secretary problem when the graph allows for multiple edges. In that case, given the single secretary problem with $n$ adversarial weights, we create a multi-graph with just two vertices $u$ and $v$ and a set of $n$ parallel edges between $u$ and $v$, with each edge assigned exactly one of these $n$ adversarial weights. Note, that this multi-graph has girth $2$. We will need to work harder to reduce the single secretary problem to the graphical matroid problem on simple graphs, where we will even assume that the simple graph has an arbitrarily large girth.

For said reduction, we define below three different measures of performance of a secretarial algorithm. Given an algorithm $a \in \mathcal{A}$ and an instance $I$ of the matroid secretary problem, we abuse the notation from the previous section and denote by $a(I)$ the set of elements of the matroid output by $a$ on instance $I$. Let also $opt(I)$ be the optimal solution to the problem on the instance $I$. Given any subset $X$ of the matroid's elements, we denote by $w(X)$ the sum of the weights of all elements in $X$.

Let $\Pi$ be a matroid secretary problem with matroid which has $n \in \nats$ elements. Given a finite set of numbers (weights) $W \subset \R_{+}$, let $D$ be a probability distribution over the set $W^n$. We will later on instantiate $\Pi$ with the single secretary problem and with the graphical matroid secretary problem.

Let $a \in \mathcal{A}$ be a (randomized) algorithm for problem $\Pi$. Given an instance $I$ of problem $\Pi$, let $S_a(I)$ be a random variable such that $S_a(I) = 1$ if algorithm $a$ output the optimal solution on instance $I$, that is, when $a(I) = opt(I)$; and $S_a(I) = 0$ if $a(I) \not = opt(I)$. Then we have that $\Exp[S_a(I)] = \Prob[a(I) = opt(I)]$; not that the probability here is over the internal randomness of algorithm $a$ and over the randomness in the random order of the elements of the matroid.

For the convenience of this subsection, it will be useful to define an algorithm $a$ as having a performance guarantee $\alpha \leq 1$ of 
\begin{itemize}
 \item type 1 if $\Exp_{I \sim D}[\Exp[S_a(I)]] \geq \alpha$. \\
   In words, type 1 performance guarantee is the expected probability of choosing the optimum by the algorithm, where the expectation is over the random choice $I \sim D$.
 \item type 2 if $\Exp_{I \sim D}[w(a(I))/w(opt(I))] \geq \alpha$. \\
   In words, type 2 performance guarantee is the expected value of the algorithm's weight divided by the optimum weight. Note that the expectation here is not only over the random choice $I \sim D$, but also over the random choices of the algorithm $a$ and over the random choice of the matroid's elements' order. 
 \item type 3 if $\Exp_{I \sim D}[w(a(I))] \geq \Exp_{I \sim D}[w(opt(I))] \cdot \alpha$. \\
 In words, type 3 performance guarantee is the algorithm's expected weight across all instances divided by the expected optimum weight. Note that the first expectation here is not only over the random choice $I \sim D$, but also over the random choices of the algorithm $a$ and over the random choice of the matroid's elements' order. The second expectation is only over the random choice $I \sim D$.
\end{itemize}

We first reformulate Theorem \ref{thm:exists_Distr_W} in terms of the type 1 performance guarantee as Lemma \ref{lemma:Hard_Distr_1_Secretary}.

\begin{lemma}\label{lemma:Hard_Distr_1_Secretary}
    For any $\epsilon > 0$, $\rho > 1$, there exists a finite distribution $D$ of instances of weighted single secretary such that no algorithm's type $1$ performance guarantee is $\geq \frac 1 e + \epsilon$. For any instance $I$ of single secretary in the support of $D$, any two weights in $I$ differ by at least a factor of $\rho$.  
\end{lemma}

We now extend the result of Lemma \ref{lemma:Hard_Distr_1_Secretary} to the type 2 performance guarantee. This proof will make use of the guarantee depending on $\rho$ provided by Lemma \ref{lemma:Hard_Distr_1_Secretary} to bound the weight attained by an algorithm when it fails to select the maximum.
\begin{lemma}\label{lemma:Hard_Distr_2_Secretary}
    For any $\epsilon > 0$, there exists a finite distribution $D$ of instances of weighted single secretary such that no algorithm's type $2$ performance guarantee is $\geq \frac 1 e + \epsilon$.  
\end{lemma}
\begin{proof}
    Apply Lemma \ref{lemma:Hard_Distr_1_Secretary} with parameters $\epsilon' = \frac \epsilon 2$ and $\rho' = \frac 2 \epsilon$. We then have a finite distribution $D$ of instances of weighted single secretary such that no algorithm's probability of choosing the maximum is $\geq \frac 1 e + \frac \epsilon 2$, and in any instance in the support of $D$, any two weights differ by at least a factor of $\frac 2 \epsilon$.

    Given any algorithm $a \in \mathcal A$, we would like to show that the type 2 performance guarantee of $a$ on $I$ is less than $\frac 1 e + \epsilon$. For any instance $I$ in the support of $D$, let $s(I)$ be the probability that $a$ will select the maximum when $a$ is run on instance $I$. It follows that with probability $s(I)$, $w(a(I))$ is equal to $opt(I)$, while with probability $1 - s(I)$, $w(a(I))$ is at most $\frac \epsilon 2 opt(I)$, as any other element in $I$ must differ from the maximum by at least a factor of $\frac 2 \epsilon$. We then have that the expected value of $w(a(I))$ over the choice of random order is at most $s(I)opt(I) + (1 - s(I)) \frac \epsilon 2 opt(I) \leq (s(I) + \frac \epsilon 2)opt(I)$, and so the expected value of $\frac {w(a(I))} {w(opt(I))}$ is at most $s(I) + \frac \epsilon 2$.

    It follows that $a$ has a type $2$ performance guarantee of at most
    \begin{align*}\mathbb E_{I \sim D} \left[\frac {w(a(I))} {w(opt(I))}\right] \leq \mathbb E_{I \sim D} \left[s(I) + \frac \epsilon 2\right] = \mathbb E_{I \sim D} \left[s(I)\right] + \mathbb E_{I \sim D}\left[\frac \epsilon 2\right] = \mathbb E_{I \sim D} \left[S_a(I)\right] + \frac \epsilon 2,\end{align*}
    noting that the expected value of $S_a(I)$ over the choice of random order is $s(I)$. We then have that the term $\mathbb E_{I \sim D} \left[S_a(I)\right]$ is simply the type $1$ performance guarantee of $a$, which is less than $\frac 1 e + \frac \epsilon 2$, and so it follows that $a$ has a type $2$ performance guarantee of at most $\frac 1 e + \frac \epsilon 2 + \frac \epsilon 2 = \frac 1 e + \epsilon$.
\end{proof}

The following lemma then shows an equivalence in hardness between the type 2 and type 3 performance guarantees. The idea of its proof is to reweight the distribution $D$ by making the probability of each instance being chosen inversely proportional to its optimum weight.
\begin{lemma}\label{lemma:type2_type3}
    For any matroid secretary problem $\Pi$, there exists a bijection $B$ between finite distributions of instances of $\Pi$ such that if $D' = B(D)$, then there exists an algorithm with a type $2$ performance guarantee of $r$ on $D$ iff there exists one with a type $3$ performance guarantee of $r$ on $D'$.
\end{lemma}
\begin{proof}
    Note that all sums over instances $I$ are taken over $I$ in the support of $D$.
    
    We will reweight $D$ to obtain $D' = B(D)$. For each instance $I$ in the support of $D$, define $p(I)$ to be the probability that $I$ is chosen from $D$. Then, define $C = \sum_I \frac{p(I)} {w(opt(I))}$. We now define $D'$ by saying that $I$ is selected from $D'$ with probability equal to $\frac {p(I)} {C \cdot w(opt(I))}$. By the definition of $C$ this is a valid distribution. Further note that this is a bijection, as $D$ can be obtained uniquely from $D'$ by reweighting the probabilities of instances in the support of $D'$ proportionally to their optimum weight.

    We can then argue that for any algorithm $a$, the expected value of its weight on $D'$ is
    \begin{align*}
        \mathbb E_{I \sim D'}\left[w(a(I))\right] = \sum_I \frac {p(I)} {C \cdot w(opt(I))} w(a(I)) = \frac 1 C \sum_I p(I) \frac {w(a(I)} {w(opt(I))} = \frac 1 C  \mathbb E_{I \sim D}\left[\frac {w(a(I))} {w(opt(I))}\right],
    \end{align*}
    while the expected value over $D'$ of the optimum weight is
    \begin{align*}
        \mathbb E_{I \sim D'}\left[w(opt(I))\right] = \sum_I \frac {p(I)} {C \cdot w(opt(I))} w(opt(I)) = \frac 1 C \sum_I p(I) = \frac 1 C.
    \end{align*}
    It then follows that the type 3 performance guarantee of $a$ on $D'$ is
    \begin{align*}
        \frac {\mathbb E_{I \sim D'}\left[w(a(I))\right]} {\mathbb E_{I \sim D'}\left[w(opt(I))\right]} = \frac {\frac 1 C  \mathbb E_{I \sim D}\left[\frac {w(a(I))} {w(opt(I))}\right]} {\frac 1 C} = \mathbb E_{I \sim D}\left[\frac {w(a(I))} {w(opt(I))}\right],
    \end{align*}
    which is simply the type 2 performance guarantee of $a$ on $D$. Thus, for any algorithm $a$, its type 2 performance guarantee on $D$ is the same as its type 3 performance guarantee on $D'$, from which it follows that there exists an algorithm with a type $2$ performance guarantee of $r$ on $D$ iff there exists one with a type $3$ performance guarantee of $r$ on $D'$, as to show either direction we can simply take the same algorithm.
\end{proof}

We now move to the portion of the proof that extends the hardness of the single secretary problem to the graph setting. As part of this reduction, we first show the following lemma, demonstrating the existence of bipartite graphs of high girth with the key property that one part is both much larger than the other and has high degree. This construction proceeds by applying the probabilistic method to derive a graph with our desired properties by removing edges from a high-girth Ramanujan graph.

\begin{lemma}\label{lemma:Finite_Ramsey}
    For any $d, g, t \in \N$ with $d \geq 2$, $t \geq 4$, there exist $m, n \in \nats$ such that 
$\frac n m \geq t$ and there exists a graph $G'$ of girth at least $g$, whose vertices can be partitioned into sets $A$ and $B$, such that $|A| = m$, $|B| = n$, all edges are between $A$ and $B$, and each vertex in $B$ has degree at least $d$.
\end{lemma}

\begin{proof}
Note that in the following proof, we will assume at certain points that $t, d$ are sufficiently large -- this is done without loss of generality as a graph which satisfies the statement of the lemma for larger $t, d$ also satisfies it for respectively smaller $t, d$.

\noindent
    Our construction will be based on using Ramanujan graphs with large girth. Theorem 4.2.2. in the book \cite{DavidoffSV2003} (see page 114) states the following:

\begin{quote}
Let $p, q$ be distinct, odd primes, such that $q > 2\sqrt{p}$. There exists a graph, denoted $X^{p,q}$, which is $(p+1)$-regular, connected and Ramanujan. Moreover, if $\left(\frac{p}{q}\right) = 1$, then graph $X^{p,q}$ has $\frac{q(q^2-1)}{2}$ vertices and its girth is at least $2 \log_p (q)$.
\end{quote}

We note that $\left(\frac{p}{q}\right)$ above is the Legendre symbol and it means that $\left(\frac{p}{q}\right) \equiv p^{\frac{q-1}{2}} \Mod{q}$. Let us take the smallest prime $p$ such that $p+1 \geq 4td$. By the Bertrand-Chebyshev Theorem, said prime fulfills  $8td - 2 \geq p \geq 4td - 1$. For any prime number $q$, $\left(\frac{p}{q}\right) = 1$ means that $p$ is a quadratic residue modulo $q$. 

Using the well known law of quadratic reciprocity we can say that $\left(\frac{p}{q}\right) = \left(\frac{q}{p}\right)$ if 
$p \equiv q \Mod{4}$. Then we note that if $q \equiv 1 \Mod{p}$, then $\left(\frac{q}{p}\right) = \left(\frac{1}{p}\right) = 1$. Thus we can consider the arithmetic sequence defined as those numbers equal to $p \Mod{4}$ and $1 \Mod{p}$ (said arithmetic sequence exists as $p, 4$ are coprime). By Dirichlet's Theorem on Primes in Arithmetic Progressions, there must be infinitely many primes $q$ in this sequence and such $q$ then satisfy $\left(\frac{p}{q}\right) =1$.

This means that given prime $p$, we can now take a sufficiently large prime $q$ such that $\left(\frac{p}{q}\right) = 1$, $q > 2\sqrt{p}$ and $q \geq (8td)^{g/2}$. Then, there exists the Ramanujan graph $X^{p,q} = (V,E)$, in which $|V| = N = \frac{q(q^2-1)}{2}$ and the degree of each vertex is $p+1 \geq 4td$.

To construct the graph $G = (V,E')$ on the set of vertices $V$, we select $A$ randomly -- independently for each vertex $u \in V$, we choose it for $A$ with probability $\frac{1}{2t}$. $B$ is then simply $V \setminus A$; for any vertex $u \in B$, we put into $E'$ only the edges between vertex $u$ to and a vertex in $A$. That is, $E' = \{\{u,v\} \in E : u \in B, v \in A\}$.

Let the random variable $S_u = 1$ if $u \in A$ and $S_u = 0$ otherwise. Then for $S  = \sum_{u \in V} S_u$ we have that $\Exp[S] = N/2t$. It follows that for any $\delta > 0$, by a Chernoff bound we have that 
$$ 
  \Pr\left[|A| \geq (1+\delta) \cdot \frac{N}{2t}\right] \leq 
  \exp\left(\frac{-\delta^2 N}{2(2+\delta)t}\right) = \tau_0 \, .
$$
Thus, letting $E_1$ be the event that $A$ is more than a $\frac {1 + \delta} {2t}$ proportion of the vertices in $G$, we have that $\Pr[E_1] \leq \tau_0$. Notably, $\tau_0$ becomes arbitrarily small as $N$ goes to $\infty$ -- this is important because we would like to minimize the size of $A$. We will take $\delta = \frac 1 2$, so that $\tau_0$ bounds the probability that $A$ is more than a $\frac 3 {4t}$ proportion of the vertices in $G$, then assume that $N$ is sufficiently large that $\tau_0$ is at most $\frac 1 {10}$. We will also assume that $t$ is sufficiently large that $\frac 3 {4t} \leq \frac 1 {18}$.

Moving to the condition that each vertex has degree at least $d$, for each vertex $u \in V$ define the random variable $Y_u$ to be equal to the number of edges in $X^{p, q}$ between $u$ and a vertex $v \in A$ -- when $u \in B$, this is equal to the degree of $u$ in $G$. We can see that $Y_u = \sum_{v | (u, v) \in E} S_v$; because the random variables $S_v$ are independent, we can apply a Chernoff bound as before. Note that as $u$ has degree at least $4td$ in $X^{p, q}$, we have that $\E\left[Y_u\right] \geq \frac {4td} {2t} = 2d$. It then follows that
\begin{align*}
    \Pr\left[Y_u \leq d\right] = \Pr\left[Y_u \leq \left(1 - \frac 1 2\right)2d\right] \leq \exp\left(\frac {-2d\left(\frac 1 2\right)^2} 2\right) = \exp\left(-\frac d 4\right)
\end{align*}
Assuming that $d$ is sufficiently large ($d \geq 12$ suffices), we have that $\Pr\left[Y_u \leq d\right] \leq \frac 1 {20}$. It follows that the expected proportion of vertices $u$ such that $Y_u \leq d$ is at most $\frac 1 {20}$. We would then like to condition on the event $E_1$ not occurring -- because $E_1$ occurs with probability at most $\frac 1 {10}$, $E_1$ does not occur with probability at least $\frac 9 {10}$; this then implies that given that $E_1$ does not occur, the expected proportion of vertices $u$ with $Y_u \leq d$ is at most $\frac {10} 9 \cdot \frac 1 {20} = \frac 1 {18}$.

Given this, define a "bad" vertex to be a vertex $u$ such that either $u \in A$ or $Y_u \leq d$. Combining our previous bounds, given that $E_1$ does not occur, the expected proportion of bad vertices in $G$ is at most $\frac 3 {4t} + \frac 1 {18} \leq \frac 1 {18} + \frac 1 {18} = \frac 1 9$. We now modify $G$ by removing any vertices $u$ in $B$ such that $Y_u \leq d$. We then have that all vertices in $B$ have degree greater than $d$, and that $|B|$ in expectation is at least $\frac 8 9 N$. Now note that $\frac 8 9 N \geq \frac 3 4 N$, while given that we are conditioning on $E_1$ not occurring, $|A| \leq \frac 3 {4t} N$ -- this means that in expectation, the size of $B$ is at least $t$ times the size of $A$. Therefore, by the probabilistic method, there exists a realization of $G$ such that $\frac {|B|} {|A|} \geq t$. We have already shown that all vertices in $B$ have degree at least $d$, and the girth of $G$ must be at least the girth of $X^{p, q}$ as we have only removed edges, so we are done.
\end{proof}

We now apply the graph provided by Lemma \ref{lemma:Finite_Ramsey} to extend the hardness of a distribution of instances of the single secretary problem to hardness of matroid secretary on a high-girth graph.

\begin{lemma}\label{lemma:type3_type3}
    If there exists a finite distribution $D$ of instances of weighted single secretary such that no algorithm has a type $3$ performance guarantee $\geq r$ on $D$, then for any $\epsilon > 0, g \in \N$ there exists a graph $G$ with girth $\geq g$ and a finite distribution $D_G$ of instances of weighted graph secretary on $G$ such that no algorithm has a type $3$ performance guarantee $\geq r + \epsilon$ on $D'$.  
\end{lemma}
\begin{proof}
 Assume that the distribution $D$ is over a finite set of instances of the single secretary problem with $d$ items, where $d$ is the degree of each vertex $u \in B$. Let the bipartite graph $G = (V' = A \cup B,E'')$ be as constructed in Lemma \ref{lemma:Finite_Ramsey}, with $d, g$ as already defined; $t$ will be chosen later.

 We construct the distribution $D_G$ of the instances of the weighted graph secretary problem on $G$ in the following way. For each vertex $u \in B$ we independently sample an instance $I$ from the distribution $D$, and then we assign the $d$ weights in the instance randomly to the edges incident to vertex $u$.

 Suppose now that we have an algorithm $a$ for the graph matroid secretary problem on $G$ with type 3 performance guarantee $\geq r + \epsilon$ on $D_G$. We will show that it leads to an algorithm $b$ for the single secretary problem with type $3$ performance guarantee $\geq r$ on $D$.

 Algorithm $b$ will act as follows on the weighted single secretary problem: it will simulate algorithm $a$ on $D_G$; however, for a single vertex $v$ chosen uniformly at random from $B$, rather than itself sampling an instance from $D$ in order to determine the weights on edges incident to $v$, it will randomly shuffle said edges, then take the weight of the $j$\th edge to be the $j$\th weight presented in the weighted single secretary problem. It will then mimic the action of $a$ on $v$: when algorithm $a$ takes an edge incident to $v$ for the first time, algorithm $b$ will take the corresponding item in the weighted single secretary game. Any further edges incident to $v$ that are taken will be ignored. This is possible because the weights from the single secretary game are presented by $b$ to $a$ in the same order that they are presented to $b$, and $a$ then immediately decides whether or not to take said edge, meaning that $b$ can then take the corresponding item if necessary before having to reveal the next item's weight.

 The key value of algorithm $b$'s choice of weights for edges incident to $v$ is that from the perspective of algorithm $a$, the weights for $v$ are being selected in the exact same manner as those for all other vertices in $B$; i.e. by independently sampling from $D$. We can therefore equivalently think of $v$ as having been chosen uniformly at random from $B$ after the execution of $a$ on an instance sampled from $D_G$.

 We now proceed to analyze the performance of $a$. Below, we will use $a(G)$ to refer to the set of edges taken by algorithm $a$, and furthermore for any vertex $u$ we will use $a(u)$ to refer to the set of edges incident to $u$ taken by algorithm $a$. Let $T = (V', a(G))$ be the spanning tree for graph $G$ resulting from $a$'s selection.
 
 Let the vertices in the set $Good = \{u \in B: deg_{T}(u) \leq 1\}$ be called good vertices -- these vertices will correspond to ``correct solutions'' to the single secretary sub-problems on such vertices of set $B$. Similarly define $Bad = B \setminus Good$.
 
 We first observe that $k = |Bad| \leq m-1$. This follows by the fact that the subgraph of $T$ induced on the vertex set $Bad \cup A$ is a tree, so the total number of edges incident on vertices in $Bad$ is at most $m+k-1$. Then, each vertex in $Bad$ has degree at least $2$ in $T$, so $ m+k-1 \geq 2k$, implying that $k \leq m-1$.

 Now, let $X = \Exp[w(opt(G))]$, where the expectation is over the choice of $n$ independent samples from distribution $D$ -- one for each vertex $u \in B$. 
 Let us observe now that under any fixed assignment of the weights to all edges in graph $G$, the following is a lower bound on the weight of the maximum weight spanning tree: take the maximum weight edge incident to each of the vertices in $B$.
 
 This implies that we can lower bound the expected optimum as follows:
 $$
   X = \Exp[w(opt(G))] 
   \geq n \cdot \Exp_{I \sim D}[w(opt(I))] \, ,
 $$ where $I$ denotes the random instance of the single secretary problem.

 We also observe that for any edge $e = (u,v)$ for some $u \in B, v \in A$, we have that the expected weight of $e$ is at most the expected largest weight of an edge incident to $u$, which is then the expected optimum of an instance sampled from $D$. Thus, $\Exp_{I \sim D}(w(e)) \leq \Exp_{I \sim D}[w(opt(I))]$; note that $w(e)$ is a random variable here.

 Below, the expectation $\Exp[w(a(G))]$ is taken over the sampling of instances from $D$, as well as over the randomness of algorithm $a$ and the randomness in the arrival order.

 As algorithm $a$ obtains a type 3 performance guarantee of $r + \epsilon$ on $D_G$, it must be that $\Exp[w(a(G))]/X \geq r+\epsilon$. Given this, we will show that algorithm $b$ obtains a type $3$ performance guarantee of $r$ on $D$.
 
 We have that
 $$
   \frac{1}{X} \cdot \Exp[w(a(G))] = \frac{1}{X} \cdot \Exp\left[\sum_{u \in Good} w(a(u))\right]
   \,\, + \,\, \frac{1}{X} \cdot \Exp\left[\sum_{u \in Bad} w(a(u))\right] \, ,
 $$
 By the previous argument, we have that the total number of edges incident on bad vertices in the tree $T$ is at most $2m-2$. This then implies that
 $$
  \Exp\left[\sum_{u \in Bad} w(a(u))\right] = 
  \Exp\left[\sum_{u \in Bad} \sum_{e \in a(u)} w(e)\right] \leq (2m-2) \cdot \Exp_{I \sim D}[w(opt(I))] \, .
 $$ which by 
   $X \geq n \cdot \Exp_{I \sim D}[w(opt(I))]$ implies that 
 $
   \frac{1}{X} \cdot \Exp\left[\sum_{u \in Bad} w(a(u))\right] \leq \frac{2m-2}{n} \leq \frac{2}{t} \, ,
 $ because $\frac n m \geq t$. Putting these estimates together, we see that
 $$
 r + \epsilon \leq \frac{1}{X} \cdot \Exp[w(a(G))] \leq \frac{1}{X} \cdot \Exp\left[\sum_{u \in Good} w(a(u))\right]
   \,\, + \,\, \frac{2}{t} \, .
 $$
 If we now constrain $t$ to satisfy $\epsilon > \frac 2 t$, we simplify this inequality to $r \leq \frac 1 X \cdot \Exp \left[\sum_{u \in Good} w(a(u))\right]$. By then applying the inequality $X \geq n \cdot \Exp_{I \sim D}[w(opt(I))]$ and rearranging, we obtain that 
 $
   \Exp\left[\frac{1}{n} \sum_{u \in Good} w(a(u))\right] \geq r \cdot \Exp_{I \sim D}[w(opt(I))] \, 
 $.
 
 Finally, recalling the operation of algorithm $b$, define $s(u)$ to be the weight obtained by algorithm $b$ given that $u$ is chosen to be $v$. As $v$ can be interpreted as having been chosen randomly after the execution of $a$, the expected weight $\Exp_{I \sim D}[w(b(I))]$ obtained by algorithm $b$ is equal to the average of the expected value of $s(u)$ over all vertices $u$. Furthermore, for vertices $u \in Good$, the weight $s(u)$ is simply the sum $w(a(u))$ of the weights of the either $0$ or $1$ edges incident to $u$ that $a$ takes; for vertices $u \in Bad$, we can lower bound $s(u)$ by $0$. 

 Therefore, we can finally argue that
 \begin{align*}
     \Exp_{I \sim D}[w(b(I))] = \Exp\left[\frac{1}{n} \sum_{u \in V} s(u)\right] &\geq \Exp\left[\frac{1}{n} \sum_{u \in Good} w(a(u))\right] \geq r \cdot \Exp_{I \sim D}[w(opt(I))].
 \end{align*}
 It follows immediately that algorithm $b$ obtains a type $3$ performance guarantee of $r$ on $D$. This contradicts the assumption on $D$, and so no such $a$ can exist as desired.
 \end{proof}

We now finally prove Theorem \ref{lowerboundthm}, which states that there does not exist an algorithm for the graphic matroid secretary problem on graphs of girth $\geq g$ that obtains competitive ratio less than $e$, by proving the below theorem, noting that the type $2$ performance guarantee is the inverse of the usual competitive ratio.
\begin{theorem}
    For any $\epsilon > 0, g \in \N$, there exists a graph $G$ with girth $\geq g$ and a finite distribution $D$ of instances of weighted graph secretary on $G$ such that no algorithm's type $2$ performance guarantee is $\geq \frac 1 e + \epsilon$.
\end{theorem}
\begin{proof}

We first apply Lemma \ref{lemma:Hard_Distr_2_Secretary} with parameter $\frac \epsilon 2$ to obtain a finite distribution $D_1$ of instances of single secretary such that no algorithm achieves a type 2 performance guarantee $\geq \frac 1 e + \frac \epsilon 2$ on $D_1$. We then apply Lemma \ref{lemma:type2_type3} on $D_1$ to obtain a finite distribution $D_2$ of instances of single secretary such that no algorithm achieves a type 3 performance guarantee $\geq \frac 1 e + \frac \epsilon 2$ on $D_2$. Following that, we apply Lemma \ref{lemma:type3_type3} on $D_2$ with parameters $r' = \frac 1 e + \frac \epsilon 2$, $\epsilon' = \frac \epsilon 2$, and $g' = g$ to obtain a graph $G$ with girth $\geq g$ and a finite distribution $D_3$ of instances of weighted graph secretary on $G$ such that no algorithm achieves a type 3 performance guarantee $\geq \frac 1 e + \frac \epsilon 2 + \frac \epsilon 2 = \frac 1 e + \epsilon$ on $D_3$. Finally, we apply Lemma \ref{lemma:type2_type3} in reverse on $D_3$ to obtain a finite distribution $D_4$ of instances of weighted graph secretary on $G$ such that no algorithm achieves a type 3 performance guarantee $\geq \frac 1 e + \epsilon$ on $D_4$. $G$ and $D_4$ are then as desired.
\end{proof}

\section{Tight algorithm for high girth graphs: proof of Theorem~\ref{highgirththm}}
\label{sec:high-girth}
In this section, we present an algorithm whose competitive ratio decreases with the girth of the graph it executes on. Specifically, let $\mathcal{G}_{g}$ denote the family of graphs with girth $\ge g$. When executed on a graph $G \in \mathcal{G}_{g}$, the algorithm achieves a competitive ratio of $ \frac{e}{x_g}$, where $x_{g}$ is the unique root in $[0, 1]$ of the polynomial $f(x) = -x^{g-1} - x + 1$. 
Here, we would like to note two facts~\footnote{We provide proofs of these facts in the Appendix, Lemma~\ref{lem:f}.}. First, it is easy to show that $\lim_{g \rightarrow \infty} x_{g} = 1$, implying that in the limit, the competitive ratio approaches $e$. Therefore, in the limit, the algorithm matches the impossibility result of Theorem~\ref{lowerboundthm}. Second, the convergence of $x_{g}$ is almost inversely linear in $g$, yielding values close to $\frac{1}{e}$ even for small $g$. In fact, for $g = 5$, the resulting competitive ratio is $<3.76$, which improves on the result obtained in Theorem~\ref{mainthm}.

\paragraph{Overview.}
The pseudocode of the algorithm is provided in Figure~\ref{alg:high_girth}. We assume that the algorithm is given an additional parameter $g$ as input, representing the minimal girth of the graph on which it is executed. The algorithm consists of two phases. During the first $\floor{\frac{m}{e}}$ steps, edges are only observed, and no actions are taken. When an edge $e_t$ arrives at time $t$, where $t \in \left[\floor{\frac{m}{e}} + 1, m \right]$, the algorithm computes the maximum spanning forest $T_{t}^{opt}$ of all edges seen so far. If $e$ does not belong to the spanning forest, it is rejected. Otherwise, let $(u, v) = e$ denote the direction of $e$ pointing to the vertex with the minimum label in $e$'s connected component in the spanning forest. If vertex $u$ has not \textit{selected} an edge by time $t$, $u$ selects edge $e$, and $e = (u,v)$ is added to an auxiliary set $A^{\ast}$. Otherwise, $e$ is rejected. Finally, two rules determine whether $e$ is added to the set of accepted edges, $A$. Rule (A): If adding $e$ to $A^{\ast}$ does not create a cycle, $e$ is added to $A$ with probability $\rho := x_{g}$. Rule (B): If adding $e$ to $A^{\ast}$ creates a cycle, $e$ is added to $A$ only if adding it does not create a cycle in $A$.

The rules underpin the following idea. The set $A^{\ast}$ is a superset of the set of accepted edges. Since adding $(u,v) = e$ requires that vertex $u$ selects $e$, edges in $A^{\ast}$ are always selected by a vertex, and each vertex can select at most one edge. This gives the set $A^{\ast}$ a specific structure: it forms a graph where each connected component can contain at most one cycle. Once the edge $e$ is added to $A^{\ast}$, there is an opportunity to add it to $A$. We can distinguish between two cases. In the first case, $e$ closes a cycle in $A^{\ast}$. In this situation, all edges that could form a cycle with $e$ have already been considered by the algorithm, allowing it to deterministically decide, based on past actions, whether adding $e$ to $A$ would create a cycle. In the remaining case, where adding $e$ to $A^{\ast}$ does not create a cycle, there is a possibility that some future edges of higher weight might form a cycle with $e$. To prevent $e$ from blocking future edges, the algorithm allows a probability of $1-\rho$ for not adding $e$ to $A$.

\paragraph{Analysis.}
We divide the analysis into two parts. First, we prove that for any edge $e \in OPT$ there is at least $\frac{1}{e}$ probability of adding $e$ to the set $A^{\ast}$. Next, we prove that conditioned on the fact that $e$ belongs the set $A^{\ast}$ the probability of $e$ being included is at least $x_{g}$.

\begin{algorithm}[!ht]
  Let $E'$ be the first $\floor{\frac m e}$ edges \hfill \Comment{set of observed edges}
  $\accset^{\ast} \gets \emptyset$ \hfill\Comment{superset of accepted edges}
  $\accset \gets \emptyset$ \hfill\Comment{set of accepted edges}
  
  $\forall v\in V$: $selected(v)\gets null$ \hfill\Comment{every vertex initially selects no edge}
  \For{$t \in \cbr{\floor{\frac m e} + 1, \dots, m}$} {
    Let $e_t$ be the edge arriving in time $t$\\
    Add $e_t$ to $E'$ \\
    $\optmatch_t \gets $ maximum weighted 
    forest on $G[E']$\\
    \If{$e_{t} \in \optmatch_{t}$} { 
    $e_{t} \gets (u, v)$ \hfill\Comment{directing $e_{t}$ towards the vertex with the smallest label in $e_{t}$'s connected component}
    \If{selected(u) $= null$}
    {
      $selected(u) \gets e_{t}$ \\
      add $e_{t}$ to $A^{\ast}$\\
      \If{$e_{t}$ forms a cycle in $A^{\ast}$} {
        Add $e_t$ to $\accset$ if $e_{t}$ does not create a cycle in $\accset$
      }
      \Else{
        Add $e_t$ to $\accset$ with probability $\rho := x_{g}$
      }
    }
  }
  }
  \caption{Algorithm for graphic matroid secretary for graphs of girth at least $g$.\label{alg:high_girth}}
\end{algorithm}

\begin{lemma}\label{lem:graphical-ext}
Let $e = (u,v)$ be an edge in $OPT$. It holds that $\Pr[e \in A^{\ast}] \ge \frac{1}{e} - \frac{2e}{m}$.
\end{lemma}
\begin{proof}
Fix a time \( t \) at which the edge \( e \) arrives, i.e., \( e_{t} = e \). Since the input is randomly permuted, we have \( \Pr[e_{t} = e] = \frac{1}{m} \), and the event \( e_{t} = e \) is independent of the order of edges arriving before time \( t \).

The set of maximum-weighted forests forms a basis system of a matroid, which implies that if \( e \in OPT \), then \( e_{t} \in T_{t}^{OPT} \). Let \( (u, v) \) represent the direction of \( e_{t} \) in \( T_{t}^{OPT} \). Note that this direction is independent of the order in which edges arrive before time \( t \), as the direction is uniquely determined by the structure of the maximum-weighted forest, which in turn is uniquely determined by the \textit{set} of edges arriving up to time \( t \).

To add the edge \( e_{t} = (u, v) \) to \( A^{\ast} \), it must be the case that \( u \) has not selected any edge before time \( t \). We now seek to lower-bound the probability of this event. Consider step \( t - 1 \). Once again, the maximum-weighted forest \( T_{t-1}^{OPT} \) is independent of the order of edge arrivals. Assign each edge in \( T_{t-1}^{OPT} \) a direction toward the vertex in its connected component in \( T_{t-1}^{OPT} \) with the smallest label. Denote this set of directed edges as \( \overrightarrow{E_{t-1}} \). We first prove the following fact.
\begin{fact}
The out-degree of vertex $u$ in the set of edges $\overrightarrow{E_{t-1}}$ is at most $1$.
\end{fact}
\begin{proof}
Assume that the out-degree of \( u \) is greater than 2. Let \( e_{1} = (u, x_1) \) and \( e_{2} = (u, x_2) \) be two distinct edges outgoing from \( u \). Since \( T_{t-1}^{OPT} \) is a forest, the vertices \( x_{1} \) and \( x_{2} \) belong to two different connected components in \( T_{t-1}^{OPT} - e_{1} - e_{2} \). Without loss of generality, let \( x_{1} \) belong to the connected component with the smallest label among all vertices in these two connected components. This implies that the smallest label of the connected component of vertex \( u \) in \( T_{t-1}^{OPT} \) cannot be found in the part where \( x_{2} \) is, effectively proving that the edge \( e_{2} \) cannot have the direction \( (u, x_2) \). This contradiction completes the proof.
\end{proof}
Next, we continue to proving the lower bound on the probability of adding $e_{t}$ to $A^{\ast}$. Recall that we are lower-bounding the probability of $u$ selecting an edge at time $t-1$. For \( u \) to select an edge in this time step, the edge from \( \overrightarrow{E_{t-1}} \) directed outward from \( u \) (if any) must appear at time \( t - 1 \), for which there is a probability of \( \frac{1}{t - 1} \). Thus, with probability at least \( 1 - \frac{1}{t - 1} = \frac{t - 2}{t - 1} \), the vertex \( u \) does not select an edge at time \( t - 1 \). Additionally, conditioning on this event still preserves independence with respect to the order of edges up to time \( t - 2 \). Therefore, by following this argument inductively from time \( t - 1 \) down to \( \floor{ \frac{m}{e}} + 1 \), we obtain:

\[
\Pr[u \text{ has not selected an edge until time } t] \geq \prod_{k = \floor{m/e} + 1}^{t - 1} \frac{k - 1}{k} = \frac{\floor{m/e} + 1}{t - 1}.
\]

This yields that $\Pr[e_{t} \in A^{\ast}]\ge \frac{\floor{m/e} + 1}{t - 1}$, and given the fact that is equally likely for $e$ to appear at any time $t \in \left[\floor{m/e} + 1, m\right]$, we get,

\[
\Pr[e \in A^{\ast}] \geq \sum_{t = \floor{m/e} + 1}^{m} \frac{1}{m}\cdot\frac{\floor{m/e} + 1}{t - 1}
\]
\[
\ge \frac{m}{e} \cdot \frac{1}{m} \sum_{t = \floor{m/e} + 1}^{m} \frac{1}{t - 1} \ge \frac{1}{e} \left( \ln\left(\frac{m - 2}{\floor{m/e}}\right) \right) \ge \frac{1}{e} - \frac{2e}{m},
\]

which concludes the proof.
\end{proof}

\begin{lemma}\label{lem:graphical-final}
Let $e$ be an edge from $A^{\ast}$. With probability at least $x_{g}$ it holds that $e \in A$.
\end{lemma}
\begin{proof}
The algorithm has two cases. In the first case, if adding \( e \) to \( A^{\ast} \) does not create a cycle, then \( e \) is added to \( A \) with probability \( \rho = x_{g} \), and thus the lemma follows.

The second case arises when adding \( e \) to \( A^{\ast} \) creates a cycle in \( A^{\ast} \). In this scenario, \( e \) is added to \( A \) only if \( e \) does not create a cycle with the other edges already in \( A \).

We first observe that in this case, \( e \) belongs to exactly one cycle in \( A^{\ast} \). Suppose, for contradiction, that there exist two cycles \( C_{1} = \left(e_{1}, \ldots, e_{p} \right) \) and \( C_{2} = \left(d_{1}, \ldots, d_{q} \right) \) such that \( e_{1} = e = d_{1} \). Here, we list the cycles by the edges they contain. If \( e_{i} \neq d_{i} \) is the first edge at which the cycles differ, this leads to a contradiction: since every vertex can select at most one edge, the two edges \( e_{i} \) and \( d_{i} \) would have to be selected by the same vertex as they are the first differing edges.

Continuing with the main proof, since \( A \subseteq A^{\ast} \), \( e \) could only create a cycle with the edges that form a cycle with \( e \) in \( A^{\ast} \). Denote the set of edges forming this cycle as \( C \). Given that the girth of \( G \) is at least \( g \), we have \( |C| \ge g \). Observe that upon adding any edge \( e' \in C \), where \( e' \neq e \), to \( A^{\ast} \), it could not create a cycle in \( A^{\ast} \) since this occurred before the arrival of \( e \). Thus, each edge \( e' \in C \), where \( e' \neq e \), was added to \( A \) with probability \( \rho \). The probability that all these edges were added to \( A \) is then \( \rho^{g-1} \), as all events are independent. Therefore, adding \( e \) to \( A \) will not create a cycle with probability \( 1 - \rho^{g - 1} \). Given that \( x_{g} \) is the root of the polynomial
\[
f(x) = -x^{g - 1} - x + 1
\]
in \( [0, 1] \), we conclude that
\[
1 - \rho^{g - 1} = 1 - x_{g}^{g - 1} = x_{g},
\]
and the lemma follows.
\end{proof}

\begin{theorem}
For any graph $G$ with girth at least $g$,
the Algorithm \ref{alg:high_girth} is $\frac{e}{x_g}$ competitive.
\end{theorem}
\begin{proof}
The fact that any edge \( e \in OPT \) is added to the set with probability \( \frac{1}{e} \cdot x_{g} \) follows immediately from Lemmas~\ref{lem:graphical-ext} and ~\ref{lem:graphical-final}. This proves that the weight of the solution set \( A \) is sufficiently large.

It remains to prove that the set \( A \) is acyclic. This, however, follows from the rules of Algorithm~\ref{alg:high_girth}, which determine when an edge is added to \( A \), as well as a similar argument to the one presented in Lemma~\ref{lem:graphical-final}. An edge \( e \) can be added to \( A \) as a result of a probabilistic event; in this case, the edge does not form a cycle upon being added to \( A^{\ast} \). Since \( A \subseteq A^{\ast} \), this edge cannot form a cycle in \( A \) either. Alternatively, the edge may be added when adding it to \( A^{\ast} \) creates a cycle. However, in such a case, the algorithm explicitly checks whether adding \( e \) to \( A \) would create a cycle and rejects any edge that does so. This completes the proof.
\end{proof}

\bibliographystyle{alpha}
\bibliography{references}

\newpage

\appendix
\section{Missing proofs}
\begin{lemma}\label{lem:f}
For a fixed integer $g \ge 2$, consider the polynomial $f(x) = -x^{g-1} - x + 1$. The following holds:\\
\noindent $\textit{(a)}$ there is exactly one root of the polynomial in $[0, 1]$ which we denote $x_{g}$.\\
\noindent $\textit{(b)}$ for sufficiently large $g$ the following are upper and lower bounds on $x_{g}$
$$1 - \frac{\ln{g}}{g - 1} \le x_{g} \le 1 - \frac{1}{g - 1}.$$
In particular $\lim_{g \rightarrow \infty} x_{g} = 1$. 
\end{lemma}
\begin{proof}
\textit{(a)} This follows immediately from analyzing the first derivative of \( f(x) \) on the interval \( [0, 1] \). We have that \( f'(x) = -(g-1)\cdot x^{g-2} - 1 \), which is less than \( 0 \) over the entire interval \( [0, 1] \). Thus, \( f(x) \) is decreasing on this interval. Given that \( f(0) = 1 \), \( f(1) = -1 \), and that \( f \) is continuous, there must be exactly one root.

\textit{(b)} Consider \( f\left(1 - \frac{\ln{g}}{g - 1} \right) \). We have that
\[
f\left(1 - \frac{\ln{g}}{g - 1} \right) = -\left(1 - \frac{\ln{g}}{g - 1} \right)^{g - 1} - \left(1 - \frac{\ln{g}}{g - 1} \right) + 1 = \left(1 - \frac{\ln{g}}{g - 1} \right)^{\frac{g - 1}{\ln{g}} \cdot \ln{g}} + \frac{\ln{g}}{g - 1}.
\]
Taking the limit as \( g \to \infty \), we obtain
\[
\lim_{g \rightarrow \infty} f\left(1 - \frac{\ln{g}}{g - 1} \right) = \lim_{g \rightarrow \infty} \left(1 - \frac{\ln{g}}{g - 1} \right)^{\frac{g - 1}{\ln{g}} \cdot \ln{g}} + \frac{\ln{g}}{g - 1}
\]
\[
= \lim_{g \rightarrow \infty} -\left( \frac{1}{e} \right)^{\ln{g}} + \frac{\ln{g}}{g - 1} = \lim_{g \rightarrow \infty} -\frac{1}{g} + \frac{\ln{g}}{g - 1} = 0^{+},
\]
and it follows that for sufficiently large \( g \), the value \( f\left(1 - \frac{\ln{g}}{g - 1} \right) \) is positive.

By a similar argument, the value \( f\left(1 - \frac{1}{g - 1} \right) \) is negative for large \( g \). 
\[
\lim_{g \rightarrow \infty} f\left(1 - \frac{1}{g - 1} \right) = \lim_{g \rightarrow \infty} \left(1 - \frac{1}{g - 1} \right)^{g - 1} + \frac{1}{g - 1}
\]
\[
= \lim_{g \rightarrow \infty} -\frac{1}{e} + \frac{1}{g - 1} = - \frac{1}{e},
\]
and the proof is complete.

\end{proof}

\end{document}